\documentclass[11pt]{article}

\usepackage[utf8]{inputenc}
\usepackage[OT4]{fontenc}
\usepackage{amsmath}
\usepackage{amsthm}
\usepackage{amsfonts}
\usepackage{amssymb}
\usepackage{caption} 
\usepackage{bbm}
\usepackage[retainorgcmds]{IEEEtrantools}
\usepackage[letterpaper]{geometry}
\usepackage{thm-restate}
\usepackage{listings}

\usepackage{tikz}
\usetikzlibrary{calc, arrows, decorations.markings, positioning, fit, 
shapes.geometric}

\theoremstyle{plain}
\newtheorem{theorem}{Theorem}[section]
\newtheorem{lemma}[theorem]{Lemma}

\newtheorem{corollary}[theorem]{Corollary}

\newtheorem{claim}[theorem]{Claim}

\theoremstyle{definition}
\newtheorem{definition}[theorem]{Definition}
\newtheorem{example}[theorem]{Example}

\theoremstyle{remark}
\newtheorem{remark}[theorem]{Remark}

\newcommand{\bbN}{\mathbb{N}}

\newcommand{\cE}{\mathcal{E}}

\newcommand{\cG}{\mathcal{G}}
\newcommand{\cH}{\mathcal{H}}

\newcommand{\cP}{\mathcal{P}}
\newcommand{\cQ}{\mathcal{Q}}

\newcommand{\cS}{\mathcal{S}}

\newcommand{\cX}{\mathcal{X}}

\newcommand{\fS}{\mathfrak{S}}
\newcommand{\ordstar}{\mathord{\star}}
\newcommand{\DHJ}{\mathrm{DHJ}}
\newcommand{\HJ}{\mathrm{HJ}}

\newcommand{\CPR}{\mathrm{CPR}}
\newcommand{\good}{\mathrm{good}}

\newcommand{\1}{\mathbbm{1}}

\newcommand{\oA}{\overline{A}}
\newcommand{\ub}{\underline{b}}
\newcommand{\uB}{\underline{B}}
\newcommand{\oE}{\overline{E}}

\newcommand{\uf}{\underline{f}}
\newcommand{\op}{\overline{p}}
\newcommand{\oP}{\overline{P}}
\newcommand{\oq}{\overline{q}}

\newcommand{\oQ}{\overline{Q}}
\newcommand{\ouQ}{\underline{\oQ}}
\newcommand{\us}{\underline{s}}
\newcommand{\ocS}{\overline{\cS}}

\newcommand{\ux}{\underline{x}}

\DeclareMathOperator{\cupdot}{\dot\cup}
\DeclareMathOperator{\val}{val}
\DeclareMathOperator{\cVal}{cVal}
\DeclareMathOperator*{\EE}{E}

\DeclareMathOperator{\Id}{Id}
\DeclareMathOperator{\Hom}{Hom}
\DeclareMathOperator{\tw}{tw}

\definecolor{DSgray}{cmyk}{0,0,0,0.7}
\definecolor{DSred}{cmyk}{0,0.7,0,0.7}

\tikzstyle{vecArrow} = [thick, decoration={markings,mark=at position
  1 with {\arrow[semithick]{open triangle 60}}},
  double distance=1.4pt, shorten >= 5.5pt,
  preaction = {decorate},
  postaction = {draw,line width=1.4pt, white,shorten >= 4.5pt}]

\title{Forbidden Subgraph Bounds for Parallel Repetition
and the Density Hales-Jewett Theorem}
\author{Jan Hązła\thanks{ETH Z\"urich, Department of Computer Science, Zurich, 
Switzerland. E-mail: {\tt jan.hazla@inf.ethz.ch}.
J\.\,H.~is supported by the Swiss National Science Foundation (SNF),
project no. 200021-132508.}
\and Thomas Holenstein\thanks{
Google, Zurich, Switzerland. E-mail: {\tt thomas.holenstein@gmail.com}.
Part of this work was done while T.\,H and A.\,R were at the Simons
Institute.
}
\and Anup Rao\thanks{
University of Washington, Seattle, USA. E-mail:
{\tt anuprao@cs.washington.edu}
}}
\date{April 19, 2016}

\usepackage[pdftex]{hyperref}

\begin{document}
\maketitle

\begin{abstract}
We study a special kind of bounds
(so called forbidden subgraph bounds, cf.~Feige, Verbitsky '02)
for parallel repetition of
multi-prover games.

First, we show that forbidden subgraph upper bounds for $r \ge 3$ provers
imply the same bounds for the density Hales-Jewett theorem for alphabet
of size $r$. As a consequence, this yields a new family of games with slow
decrease in the parallel repetition value.

Second, we introduce a new technique for proving exponential forbidden
subgraph upper bounds and explore its power and limitations.
In particular, we obtain exponential upper bounds for two-prover games
with question graphs of treewidth at most two and show that our method
cannot give exponential bounds for all two-prover graphs.
\end{abstract}

\section{Introduction}

\subsection{Multi-prover games}

An \emph{$r$-prover game} is a protocol in which $r$ \emph{provers} have a 
joint objective of making another entity, the \emph{verifier}, accept. 
The execution of such a game looks as follows:
The verifier first samples $r$ questions 
$q^{(1)}, \ldots, q^{(r)} \in Q^{(1)} \times \ldots \times Q^{(r)}$.
Those questions are sampled uniformly from some \emph{question set}
$\overline{Q} \subseteq Q^{(1)} \times \ldots \times Q^{(r)}$.

Then, she sends the questions
to the provers: the $j$-th prover receives $q^{(j)}$ and sends back
an answer $a^{(j)}$ (from a finite answer alphabet $A^{(j)}$)
that depends only on $q^{(j)}$. Finally, the verifier accepts or rejects
based on the evaluation of a \emph{verification predicate}
$V(q^{(1)}, \ldots, q^{(r)}, a^{(1)}, \ldots, a^{(r)})$.

A \emph{strategy} $(\cS^{(1)}, \ldots, \cS^{(r)})$ 
for the provers consists of $r$ functions, the $j$-th of
which maps questions to answers for the $j$-th prover. The \emph{value}
of a game is
\begin{align*}
  \val(\cG) := \max_{\cS^{(1)}, \ldots, \cS^{(r)}}
  \Pr \left[
    V \left( q^{(1)}, \ldots, q^{(r)}, \cS^{(1)}(q^{(1)}), \ldots, 
    \cS^{(r)}(q^{(r)}) \right) = 1 
  \right] \; ,
\end{align*}
where maximum is over all strategies and the probability over the uniform
choice of $q^{(1)},\ldots,q^{(r)} \in \oQ$.
A game $\cG$ is called \emph{trivial} if
$\val(\cG) = 1$.

A formal definition is provided in Section~\ref{sec:definitions}. One might
consider allowing other distributions over $\oQ$ than the uniform one.
This does not make much difference for parallel repetition, as shown
in Section~\ref{sec:non-uniform}.

\subsection{Parallel repetition}

The \emph{$n$-fold parallel repetition} $\cG^n$ of an $r$-prover game $\cG$
is a game where the verifier samples $n$ independent question tuples,
sends $n$ questions to each prover, receives $n$ answers from each prover,
and accepts if all $n$ instances of the verification predicate for $\cG$ accept.

It is easy to see that $\val(\cG^n) \ge \val(\cG)^n$, in particular
if $\cG$ is trivial, then $\cG^n$ is trivial as well.
However, since the $i$-th answer of a prover can depend on all of his questions,
as opposed to just the $i$-th one, it is possible that
$\val(\cG^n)$ attains a higher value.

In this paper we are interested in upper bounds on $\val(\cG^n)$ 
that depend only on $\oQ$ and $n$.
They are called \emph{forbidden subgraph} bounds, with the name explained in
\cite{FV02}.

Specifically, let
$\omega_{\oQ}(n) := \max_{\cG} \val(\cG^n)$, where the maximum is over
all non-trivial games $\cG$ with question set\footnote{
  Note that the number of provers $r$ is implicitly determined by $\oQ$.
}
$\oQ$.
We say that a question set $\oQ$ \emph{admits parallel repetition}
if $\lim_{n \to \infty} \omega_{\oQ}(n) = 0$. Furthermore, we will say
that $\oQ$ \emph{admits exponential parallel repetition} if there exists
$C_{\oQ} < 1$ such that
\begin{align*}
  \omega_{\oQ}(n) \le (C_{\oQ})^n \; .
\end{align*} 

\paragraph{Background}
The two-prover games in the context of theoretical computer science were
first introduced by Ben-Or, Goldwasser, Kilian and Wigderson \cite{BGKW88}.
Fortnow, Rompel and Sipser 
were the first to treat the value of two-prover repeated games and exhibit
an example with $\val(\cG^2) > \val(\cG)^2$ \cite{FRS88, FRS90}.
Extensive works on parallel repetition were produced, for a
survey we refer to \cite{Fei95} and \cite{Raz10}. Here we will only mention 
some results relevant to our theorems.

\subsection{Density Hales-Jewett theorem}

\begin{definition}[Combinatorial line]
Let $r, n \in \mathbb{N}_{>0}$ and $[r] := \{1, \ldots, r\}$. 
A \emph{combinatorial pattern} over $[r]^n$
is a string 
\begin{align*}
(b_1, \ldots, b_n) = \underline{b}  \in 
  \left([r] \cup \{\ordstar\}\right)^n \setminus [r]^n \; ,
\end{align*}
where
$\ordstar$ is a special symbol called the \emph{wildcard}. Note that a pattern
contains at least one wildcard.

For $q \in [r]$ we let $\underline{b}(q) \in [r]^n$ to be the string
formed from $\underline{b}$ by substituting all occurrences of the wildcard
with $q$.

A \emph{combinatorial line} associated with a pattern $\underline{b}$ is
the set $L(\underline{b}) := \{\underline{b}(1), \ldots, \underline{b}(r)\}$.
\end{definition}

\begin{example}
For $r = 3$, $n = 5$, an example pattern is 
$\underline{b} = 12\ordstar2\ordstar$.
The corresponding combinatorial line is $L(\ub) = \{12121, 12222, 12323\}$.
\end{example}

The Hales-Jewett theorem \cite{HJ63}
says that for any $r$, and large enough $n$,
any coloring of $[r]^n$ contains a monochromatic combinatorial line.
The density Hales-Jewett theorem
is a strengthening of this result: It states
that for any $r$ and $\mu$, and large enough $n$, any subset of $[r]^n$
of measure $\mu$ contains a combinatorial line.

These theorems have played a central role in Ramsey theory because
many other results in this field
(e.g., corners theorem~\cite{AS74}, Szemerédi's theorem~\cite{Sze75}) 
can be reduced to the density Hales-Jewett theorem.

\begin{definition}
\label{def:dhj-coeff}
For a set $S \subseteq [r]^n$ we define its
\emph{measure} as $\mu(S) := |S| / r^n$. 
We then let
$\omega^{\DHJ}_r(n)$ to be the maximum measure of a subset of $[r]^n$ that
does not contain a combinatorial line.
\end{definition}

\begin{theorem}[Density Hales-Jewett theorem, \cite{FK91}]
\label{thm:dhj}
Let $r \ge 2$. Then, 
\begin{align*}
\lim_{n \to \infty} \omega_r^{\DHJ}(n) = 0 \; .
\end{align*}
\end{theorem}

The original proof of Furstenberg and Katznelson \cite{FK91} does not give 
explicit bounds for $\omega_r^{\DHJ}(n)$. This situation was improved by 
a more recent proof by Polymath \cite{Pol12}, however their bounds are still
not primitive recursive:
They prove that $\omega_r^{\DHJ}(n)$ decreases at a rate that is related
to the inverse of the $r$-th Ackermann function of $n$. In particular,
$\omega_3^{\DHJ}(n) \le O\left( 1 / \sqrt{\log^* n}\right)$.

On the other hand, another paper by Polymath \cite{Pol10} gives the best
known density Hales-Jewett lower bounds:
$\omega_r^{\DHJ}(n) \ge \exp\left(-O(\log n)^{1 / \lceil \log_2 r\rceil}
\right)$ for $r \ge 3$.
Note that $r = 2$ is a somewhat special case with
$\omega_2^{\DHJ}(n) = \Theta(1/\sqrt{n})$ known by Sperner's theorem.

Verbitsky \cite{Ver96} used the density Hales-Jewett theorem to show
that every question set admits parallel repetition:
\begin{theorem}[\cite{Ver96}]
\label{thm:dhj-pr}
Let $\oQ$ be a $k$-prover question set of size $\left|\overline{Q}\right| = r$. 
Then,
\begin{align*}
 \omega_{\oQ}(n) \le \omega_r^{\DHJ}(n) \; .
\end{align*}
In particular, $\oQ$ admits parallel repetition.
\end{theorem}

\subsection{Our results --- equivalence of DHJ and PR}

Despite numerous works, especially concerning the two-prover case,
Theorem~\ref{thm:dhj-pr} remains the best available bound for general
parallel repetition of multi-prover games. Our main result hints that
there might be a reason for this situation.

\begin{definition} \label{def:qr}
Let $r \ge 2$. We define an $r$-prover question set
$\oQ_r \subseteq \{0,1\}^r$ of size $r$, where the $j$-th question 
contains $1$ in the $j$-th position and $0$ in the remaining
positions. In other words,
\begin{align*}
  \oQ_r := \left\{
    (q^{(1)}, \ldots, q^{(r)}): |\{j: q^{(j)} = 1\}| = 1
  \right\} \; .
\end{align*}
\end{definition}

Our main theorem is a construction that ties the existence of a
combinatorial line in a set with the parallel repetition value of a certain
game. Note that it requires at least three provers.

\begin{theorem}
\label{thm:dhj-game}
Let $r \ge 3$, $n \ge 1$ and $S \subseteq [r]^n$ with
$\mu(S) = |S| / r^n$ such that $S$ does not contain a combinatorial line. 

There exists an $r$-prover game $\cG_S$ with question set $\oQ_r$ and with
answer alphabets $A^{(j)} = 2^{[n]} \times [n]$ such that:
\begin{itemize}
  \item $\val(\cG_S) \le 1 - 1/r$.
  \item $\val(\cG_S^n) \ge \mu(S)$.
\end{itemize}
\end{theorem}

Theorem \ref{thm:dhj-game} immediately implies an inequality
complementary to Theorem~\ref{thm:dhj-pr}:

\begin{theorem}
\label{thm:pr-dhj}
Let $r \ge 3$. We have $\omega_{r}^{\DHJ}(n) \le \omega_{\oQ_r}(n)$.
\end{theorem}

Our game construction is related to another one by Feige and Verbitsky 
\cite{FV02} in the following way:
They proved that the so-called forbidden subgraph 
method is universal\footnote{
As a matter of fact, they treated only two-prover games.
However, their construction generalized to the multi-prover setting.
}
for obtaining bounds on $\omega_{\oQ}(n)$
(hence the name: forbidden subgraph bounds).
In their proof they used a result that is very similar to Theorem~\ref{thm:dhj-game}
generalized to an arbitrary question set $\oQ$:
its statement is almost the same except they assume that the set $S$ does
not contain a forbidden subgraph (instead of a combinatorial line).
As a matter of fact, our game $\cG_S$ is the same as the one from their theorem
instantiated on the question set $\oQ_r$.

Consequently, Theorem~\ref{thm:dhj-game} has a shorter proof assuming
the result of Feige and Verbitsky: One only checks that for the question set
$\oQ_r$ a forbidden subgraph in $S$ implies a combinatorial line in $S$. 
This is discussed in more detail
in Section~\ref{sec:fv02}. A self-contained proof of Theorem~\ref{thm:dhj-game}
is provided in Section~\ref{sec:pr-dhj}.

Most of the work done on parallel repetition lower bounds is based on 
two-prover examples by Feige and Verbitsky \cite{FV02} and Raz \cite{Raz11}
with little known in the multi-prover case.
By taking the largest known combinatorial line-free sets
from \cite{Pol10}, Theorem~\ref{thm:dhj-game} yields a new
family of multi-prover games with a slow decrease in the parallel 
repetition value.
In particular, it establishes the first known question sets that do not
admit exponential parallel repetetition answering
an open question asked in \cite{FV02}
in multi-prover case:
\begin{theorem}
\label{thm:multi-no-exp}
Let $r \ge 3$. The question set $\oQ_r$ does not admit
exponential parallel repetition.
\end{theorem}

Furthermore, this family is related to another aspect of the lower
bound in \cite{FV02}. They show that for any two-prover
upper bound of the form
\begin{align*}
  \val(\cG^n) \le \exp \left( - f(\epsilon) \cdot g(\left|\oA\right|) \cdot n 
  \right) \; ,
\end{align*}
where $\val(\cG) = 1-\epsilon$ and $\left|\oA\right|$ is the answer set
size, it must be $g(\left|\oA\right|) \le
O\big( \log \log \left|\oA\right| / \allowbreak
\log \left| \oA \right| \big)$.
On the other hand, the famous two-prover upper bound by Raz \cite{Raz98}
has $g(\left|\oA\right|) \ge 
\Omega\left( 1 / \log \left|\oA\right| \right)$. The analysis of our family
of games implies that if such bound is ever extended to the multi-prover setting,
it must be $g\left( \left|\oA\right| \right) \le 
O\left( \left( \log \log \left|\oA\right| \right)^{\epsilon} / 
\log \left|\oA\right| \right)$ for any $\epsilon > 0$.

Those lower bounds are discussed in Section~\ref{sec:lower-bounds}.

Finally, inspired by the proof of Theorem~\ref{thm:dhj-game},
in Section~\ref{sec:coloring} we define \emph{coloring games} and their
parallel repetition and show that it is equivalent to the (coloring)
Hales-Jewett theorem.

\subsection{Our results --- upper bounds on parallel repetition}

It remains open if all two-prover question sets admit exponential parallel
repetition.
We believe that this question, though somewhat forgotten, 
becomes more interesting in light
of Theorem~\ref{thm:multi-no-exp}. If the answer is affirmative,
it would constitute a significant difference
in behavior of two and three-prover games. On the other hand, if there exist
``hard'' two-prover sets, it might be possible to use them to obtain
non-trivial consequences in vein of Theorem~\ref{thm:pr-dhj}.

Motivated by this, we explore a new method for obtaining
exponential parallel repetition:
In Section~\ref{sec:constructing} we define a notion of a
question set $\oQ$ that is
\emph{constructible by conditioning}
(see Definition~\ref{def:constructible-hypergraphs}).
It is an inductive definition
using two graph-theoretic operations (we interpret an $r$-prover
question set $\oQ$ as the edge set of an $r$-regular, $r$-partite hypergraph).
Then, using a technique inspired by a previous paper of some of the authors 
\cite{HHM15},
in Section~\ref{sec:constructability-pr-proof} we prove:
\begin{theorem}
\label{thm:constructible-good-simple}
Let $\oQ$ be an $r$-prover question set that is constructible by conditioning.
Then, $\oQ$ admits exponential parallel repetition.
\end{theorem}

In particular, this gives a new proof of exponential parallel repetition for
\emph{free} multi-prover question sets\footnote{
The proof in \cite{CCL92} is for two provers only, but it generalizes to
many provers by using a result on hypergraphs by
Erdős \cite{Erd64}.
}
\cite{CCL92, Fei95, Pel95}, 
i.e., those, where the provers' questions are independent.

In Section~\ref{sec:treewidth-two} we present an example of what can be
achieved with this method.
We show:
\begin{restatable}{theorem}{twtwo} 
\label{thm:tw-constructible}
Every bipartite graph $G$ with treewidth at most two is constructible by
conditioning.
In particular, if $G$ is interpreted as a two-prover question set,
then it admits exponential parallel repetition.
\end{restatable}

This improves on previous work by Verbitsky \cite{Ver95} and 
Weissenberger\footnote{
As a matter of fact, the result of Weissenberger is even stronger: 
it gives an exponential upper bound
that depends only on $\val(\cG)$, but not on question or answer set size.
} \cite{Wei13} which showed exponential
parallel repetition for, respectively, trees and cycles.
We note that our proof dooes not seem to be a generalisation of 
the earlier ones, but rather a genuinely new approach.

In Section~\ref{sec:non-constructible} 
we show that
our technique of graph constructability is too weak to resolve the 
two-prover question positively: 
\begin{theorem}
\label{thm:non-constructible-simple}
There exists a bipartite graph that is not constructible by conditioning.
\end{theorem}

\subsection{Comparison with information-theoretic bounds}
\label{sec:raz-compare}

The general two-prover bound by Raz \cite{Raz98} as improved by Holenstein
\cite{Hol09} gives
\begin{align}
\label{eq:03c}
\val(\cG^n) \le \exp\left(-\Omega\left( \epsilon^3 n / \log \left|\oA\right|
\right)\right)
\end{align} 
for a game $\cG$ with $\val(\cG) = 1-\epsilon$ and answer set
$\oA = A^{(1)} \times A^{(2)}$. There are numerous other bounds of this
form (i.e., with dependence in the exponent on some
power of $\epsilon$ and possibly $\left|\oA\right|$), in particular
for projection and free two-prover games \cite{Rao11, BRR09} and
for multi-prover games with quantum entanglement or no-signaling strategies,
e.g., \cite{BFS14, CWY15}.

On the other hand, we are interested in forbidden subgraph bounds that depend 
only on the question set $\oQ$. Those two types of bounds are not 
comparable to each other. Furthermore, it seems that the respective proof 
techniques are also quite different: while for the bounds like 
\eqref{eq:03c} information theory is usually employed,
the forbidden graph bounds use more combinatorial arguments.

We do not know if an ``information-theoretic'' 
bound holds for general games with more than two
provers (it is known only for free games, see~\cite{CWY15},
and for the so called \emph{anchored games} \cite{BVY15}).
We stress that our negative result in Theorem~\ref{thm:multi-no-exp} does
\emph{not} exclude the possibility of such a bound.

\section{Preliminaries}

\subsection{Notation}

Many of our results feature two-dimensional vectors. 
We adopt the following conventions:
\begin{itemize}
  \item Most of the time we consider two dimensions corresponding to
    $n$ independent coordinates and $r$ provers. 
    Usually $n$ is meant to be large compared to $r$.
  \item We index the $n$-dimension with $i$ in the subscript 
    and the $r$-dimension with $j$ in parentheses in the superscript.
    We denote aggregation over $i$ by underline and over
    $j$ by overline. For example:
    \begin{IEEEeqnarray*}{l}
      \overline{\underline{V}} = 
      (\underline{V}^{(1)}, \ldots, \underline{V}^{(j)}, \ldots, 
      \underline{V}^{(r)}) = (\overline{V}_1, \ldots, \overline{V}_i, \ldots, 
      \overline{V}_n),\\
      \underline{V}^{(j)} = (V^{(j)}_1, \ldots, V^{(j)}_i, \ldots, V^{(j)}_n),
      \quad
      \overline{V}_i = (V^{(1)}_i, \ldots, V^{(j)}_i, \ldots, V^{(r)}_i)
    \end{IEEEeqnarray*}
  \item We call the element collections aggregated over $i$ 
    (like $\underline{v}^{(j)}$)
    \emph{vectors} and the element collections aggregated over $j$ 
    (like $\overline{v}_i$) \emph{tuples}.
\end{itemize}

For sets $A$, $B$  we sometimes denote $A \cup B$ as $A \cupdot B$ to
emphasize that $A \cap B = \emptyset$.
For an event $\cE$ we denote its indicator function by $\mathbbm{1}_{\cE}$. 
Whenever we speak of a partition of a set, we allow empty classes in the
partition.
The powerset of $X$ is denoted by $2^X$. 
In accordance with the computer science tradition, the symbol $\log$
denotes the logarithm with base two.
For a string $\ux \in X^n$ and $y \in X$ we let
$w_y(\ux) := \left|\left\{ i \in [n]: x_i = y \right\}\right|$. 

\subsection{Definitions}
\label{sec:definitions}

In this section we provide formal definitions of the most important concepts
we use.

\begin{definition}[Multi-prover games] \label{def:game}
An \emph{$r$-prover game} $\cG = (\overline{Q}, \overline{A}, V)$ 
consists of the following elements:
\begin{itemize}
\item
$\overline{Q} \subseteq Q^{(1)} \times \ldots \times Q^{(r)}$ is a finite 
\emph{question set}. 

Note that $\overline{Q}$ does not have to consist of all possible tuples.
However, we will always assume that there are no ``impossible questions'',
i.e., that for each
element of a \emph{question alphabet} $q^{(j)} \in Q^{(j)}$ 
there exists at least one question tuple $\oq \in \oQ$ with $q^{(j)}$ as its
$j$-th element.

\item
$\overline{A} = A^{(1)} \times \ldots \times A^{(r)}$ is a finite 
\emph{answer set}.
\item $V: \overline{Q} \times \overline{A} \to \{0, 1\}$ is
a \emph{verification predicate}.
\end{itemize}

A \emph{strategy} $\overline{\cS} = (\cS^{(1)}, \ldots, \cS^{(r)})$ for a game
$\cG$ is a tuple of functions, where $\cS^{(j)}: Q^{(j)} \to A^{(j)}$.

Let $\overline{q} = (q^{(1)}, \ldots, q^{(r)})$ be a random variable
sampled uniformly from $\oQ$. For a strategy $\overline{\cS}$
let $\overline{\cS}(\overline{q}) := 
(\cS^{(1)}(q^{(1)}), \ldots, \cS^{(r)}(q^{(r)}))$.

We define the \emph{value} of a game $\cG$ as
\begin{align*}
\val(\cG) := \max_{\overline{\cS}} \Pr \left[ 
  V\left(\overline{q}, \overline{\cS}(\overline{q})\right) = 1 \right] \; .
\end{align*}

We say that a game is \emph{trivial} if its value is $1$. 

We also say it is \emph{free} if
$\overline{Q} = Q^{(1)} \times \ldots \times Q^{(r)}$.
Note that in our setting this is equivalent to the property that the
provers' questions are distributed independently.
\end{definition}

\begin{definition}[Parallel repetition]
The \emph{$n$-fold parallel repetition} $\cG^n$ of an 
$r$-prover game $\cG = (\overline{Q}, \overline{A}, V)$
is another $r$-prover game $\cG^n = 
\left(\overline{\underline{Q}}, \overline{\underline{A}},
\underline{V}\right)$ where
\begin{itemize}
\item The question alphabet for the $j$-th prover 
$\underline{Q}^{(j)} := \left(Q^{(j)}\right)^n$ is the $n$-fold product of the
original $Q^{(j)}$. Consequently, the question set $\overline{\underline{Q}}$
is the $n$-fold product of $\overline{Q}$.
\item In the same way, the answer alphabet for the $j$-th prover
$\underline{A}^{(j)}$ is the $n$-fold product of $A^{(j)}$.
\item The verification predicate $\underline{V}$ accepts if and only
if all of its $n$ single instances accept:
\begin{align*}
  \underline{V}\left(\underline{\overline{q}}, \underline{\overline{a}}\right) 
  = 1 \iff \forall i \in [n]:
  V\left( \overline{q}_i, \overline{a}_i \right) = 1 \; .
\end{align*}
\end{itemize}
\end{definition}

\begin{definition}
For an $r$-prover question set $\oQ$ we define 
$\omega_{\oQ}(n) := \max_{\cG} \val(\cG^n)$, where the maximum is over all 
non-trivial games $\cG$ with question set $\oQ$.

We say that $\oQ$ \emph{admits parallel repetition} if 
$\lim_{n \to \infty} \omega_{\oQ}(n) = 0$. We say that $\oQ$ 
\emph{admits exponential parallel repetition} if there exists $C_{\oQ} < 1$
such that for every $n \in \mathbb{N}$:
\begin{align*}
\omega_{\oQ}(n) \le (C_{\oQ})^n \; .
\end{align*}
\end{definition}

An important notion in our proofs is a homomorphism of $r$-regular,
$r$-partite hypergraphs:

\begin{definition}
\label{def:homomorphism}
Let $r \ge 2$ and $\oQ \subseteq Q^{(1)} \times \ldots \times Q^{(r)}$ be an
$r$-prover question set.
 Consider the $r$-regular, $r$-partite hypergraph 
$G = (Q^{(1)}, \ldots, Q^{(r)}, \overline{Q})$. 
We will often abuse the notation by identifying
this hypergraph with $\oQ$.

Given two hypergraphs $(Q^{(1)}, \ldots, Q^{(r)}, \oQ)$
and $(P^{(1)}, \ldots, P^{(r)}, \oP)$
we say that $f = (f^{(1)}, \ldots, f^{(r)})$,
$f^{(j)}: Q^{(j)} \to P^{(j)}$ 
 is a \emph{homomorphism} from $\oQ$ to $\oP$ if
$\oq = (q^{(1)}, \ldots, q^{(r)}) \in \oQ$ 
implies $f(\oq) := (f^{(1)}(q^{(1)}), \ldots, f^{(r)}(q^{(r)})) \in \oP$.

If $f$ is a homomorphism from $\oQ$ to $\oQ$ we will just say that
$f$ is a homomorphism of $\oQ$. We denote the set of homomorphisms
from $\oQ$ to $\oP$ by $\Hom(\oQ, \oP)$.
\end{definition}

There are some homomorphisms of an $r$-prover question set $\oQ$ that
are important to us: 

\begin{definition}
\label{def:id}
We denote the \emph{identity homomorphism} as $\Id$. We will also use the 
\emph{constant homomorphism} $\1_{\oq}$ for a hyperedge $\oq \in \oQ$,
where $\1_{\oq}$ maps every hyperedge to $\oq$.
\end{definition}

\subsection{Reduction of general parallel repetition to uniform case}
\label{sec:non-uniform}

We show how parallel repetition for a question distribution
that is not necessarily uniform over a question set $\oQ$
reduces to the uniform case. The proof is taken from \cite{FV02} and is
included here for completeness.

\begin{theorem}
Let $\oQ$ be an $r$-prover question set and let $\cQ$ be a probability
distribution with support $\oQ$ such that 
$\epsilon := \min_{\oq \in \oQ} \cQ(\oq)$ and
\begin{align*}
  \alpha := \frac{\epsilon}{1/\left|\oQ\right|} = \epsilon \left|\oQ\right| \; .
\end{align*}
Furthermore, assume that a function $f: \bbN_{>0} \to [0, 1]$ is such that
for every non-trivial game $\cH$ uniform over $\oQ$:
\begin{align*}
  \val(\cH^n) \le f(n) \; .
\end{align*}
Then, for every non-trivial game $\cG$ such that its questions are
sampled according to $\cQ$ we have
\begin{align*}
  \val(\cG^n) \le \exp\left(-\alpha^2 n / 2\right) + 
f\left(\alpha n / 2\right) \; .
\end{align*}
\end{theorem}

\begin{proof}
First, note that we can write $\cQ = \alpha U_{\oQ} + (1-\alpha) \cQ'$,
where $U_{\oQ}$ is the uniform distribution over $\oQ$ and $\cQ'$ some
other probability distribution. Consequently, we can define an i.i.d.~random
binary vector $\uB = (B_1, \ldots, B_n)$ coupled with an execution of $\cG^n$
such that $B_i = 1$ if the $i$-th question is sampled from $U_{\oQ}$ and
$B_i = 0$ if it was sampled from $\cQ'$.

Consider an execution of $\cG^n$ with a modified verifier. The new verifier 
first checks if $w_1(\uB) \ge \alpha n / 2$, i.e., if the number of coordinates
with $B_i = 1$ is at least half of the expectation $\alpha n$. She accepts
if this check fails.
If the first check succeeds, the new verifier accepts if the single-coordinate
verifier accepts on all coordinates with $B_i = 1$.

Let us call this modified game $\left(\cG^n\right)'$. 
Clearly, $\val(\cG^n) \le \val(\left(\cG^n\right)')$.
Furthermore, let $\cG^*$ be a game with the same verifier as $\cG$ but uniform 
over $\oQ$. Note that $\cG^*$ is non-trivial. Observe that conditioned on
a choice of $\uB$, the game $\left(\cG^n\right)'$ is the same as
$\cG^*$ repeated $w_1(\uB)$ times. Consequently, and using Chernoff bound,
\begin{align*}
\val(\cG^n) &\le \val((\cG^n)') = \EE \left[ \val((\cG^n)') \mid \uB \right]
\le \Pr \left[ w_1(\uB) < \alpha n / 2 \right] + f(\alpha n / 2)
\\ &\le \exp( - \alpha^2 n / 2) + f(\alpha n / 2) \; .  
\end{align*}
\end{proof}

\section{Parallel Repetition Implies Density Hales-Jewett}
\label{sec:pr-dhj}

In this section we prove Theorem~\ref{thm:dhj-game}. We present a 
self-contained proof here and then we explain how it is related to
the proof from \cite{FV02} in Section~\ref{sec:fv02}.

Let $r \ge 3$, $n \ge 1$ and $S \subseteq [r]^n$ with $\mu(S) = |S| / r^n$.
We want to define a game $\cG_S$ with question set $\oQ_r$ such that:
\begin{itemize}
  \item If $S$ does not contain a combinatorial line, then
    $\cG_S$ is non-trivial.
  \item $\val(\cG_S^n) \ge \mu(S)$.
\end{itemize}

Firstly, note that there is a natural bijection between the question
tuples in $\oQ_r$ and $[r]$. Namely, we can think of the verifier as
choosing the number of a \emph{special} prover 
$a \in [r]$ u.a.r.~and sending $1$
to the special prover and $0$ to all other provers.

The answer alphabet of the game $\cG_S$ is the same for all provers:
$A^{(j)} := 2^{[n]} \times [n]$.
Upon sampling a special prover $a$ and receiving answers 
$(T^{(1)},z^{(1)}),\ldots,(T^{(r)},z^{(r)})$, 
the verifier
checks the following conditions and accepts if all of them
are met:
\begin{itemize}
\item The sets $T^{(1)},\ldots,T^{(r)}$ form a partition of $[n]$.
\item $z^{(1)} = z^{(2)} = \ldots = z^{(r)} = z$.
\item $z \in T^{(a)}$.
\item Let $\us = (s_1, \ldots, s_n)$ be the string over $[r]^n$ such that
$s_i = j$ iff $i \in T^{(j)}$.
Then, $\us \in S$.
\end{itemize}

The next claim is not actually needed for the proof.
However, it provides some intuition for the construction of the game.
\begin{claim}
If $S$ has a combinatorial line, then the game $\cG_S$ is trivial.
\end{claim}
\begin{proof}
Let $\ub$ be a pattern with its combinatorial
line $L(\ub) \subseteq S$, and fix 
a position $z \in [n]$ with $b_z = \ordstar$.
For $\sigma \in [r] \cup \{\ordstar\}$, let 
$B(\sigma) := \{i :b_i  = \sigma\}$, the set of coordinates in 
which $\ub$ equals $\sigma$.
We consider the strategy in which prover $j$ responds with
\begin{align}
\cP^{(j)}(q^{(j)}) := 
\begin{cases}
\bigl(B(j),z) & \text{if $q^{(j)} = 0$,}\\
\bigl(B(j) \cup B(\ordstar),z) & \text{if $q^{(j)} = 1$.}
\end{cases}
\end{align}
Since the sets $B(1),\ldots,B(r),B(\ordstar)$ form a partition of 
$[n]$, the verifier will always accept the first condition.
The condition $z^{(1)} = \dots = z^{(r)}$ is obviously always met.
Also $z \in T^{(a)}$ is clear, since prover $a$ responds with $(B(a) \cup
B(\ordstar))$ and $z \in B(\ordstar)$.
Finally, $\us \in S$ is also clear since 
$\us$ is exactly the pattern $\ub$ with $a$ in place of stars,
i.e., $\us = \ub(a)$ and since $\ub(a) \in L(\ub) \subseteq S$.
\end{proof}
\begin{claim}\label{claim:valOneToLine}
If the game $\cG_S$ is trivial, then $S$ has a combinatorial line.
\end{claim}
\begin{proof}
Let $\cP^{(1)}, \ldots, \cP^{(r)}$ be a
strategy for the provers that always wins.

For $q \in \{0,1\}$ and $j \in [r]$, 
we let $(T^{(j)}_{(q)},z^{(j)}_{(q)}) =: \cP^{(j)}(q)$ be the answer 
which prover $j$ gives on question $q$.
Since the verifier checks 
$z_{(0)}^{(j)} = z^{(a)}_{(1)}$ whenever $j \neq a$, we see that
$z_{(0)}^{(1)} = z^{(2)}_{(0)} = \ldots = z^{(r)}_{(0)} = 
z^{(1)}_{(1)} = \ldots = z^{(r)}_{(1)} =: z$ (note that we used $r \ge 3$).

Next, for any two $j \neq j'$, 
the sets $T^{(j)}_{(0)}$ and $T^{(j')}_{(0)}$ are pairwise disjoint.
Otherwise, if the verifier chooses $a$ 
which is different from both $j$ and $j'$, she will reject.

Furthermore, $z \notin T^{(1)}_{(0)} \cup \ldots \cup 
T^{(r)}_{(0)}$, since 
if $z \in T^{(j)}_{(0)}$, the verifier rejects if $a \neq j$.
Hence, the following defines a combinatorial pattern $\ub$:
\begin{align}
b_i := 
\begin{cases}
j & \text{if $i \in T_{(0)}^{(j)}$, for 
$j \in [r]$,}\\
\ordstar & \text{otherwise.}
\end{cases}
\end{align}

Fix now $a \in [r]$.
We show that $\ub(a) \in S$.
Suppose the verifier picks $a$ as the special prover.
Since the verifier checks that the sets
$T^{(\cdot)}$ form a partition, it must be
that prover $a$ responds with $T_{(1)}^{(a)} = [n] \setminus
\left(T_{(0)}^{(1)} \cup \ldots \cup T_{(0)}^{(a-1)} \cup 
T_{(0)}^{(a+1)} \cup \ldots T_{(0)}^{(r)}\right)$.
Since the verifier checks that the resulting string is in $S$ and
accepts, it must be that $\ub(a) \in S$.
This holds for every $a$, and thus $L(\ub) \subseteq S$.
\end{proof}

\begin{claim}\label{claim:valofPR}
The value of $\cG^n_S$ is at least $\mu(S)$.
\end{claim}
\begin{proof}
Let $T^{(j)}$ be the set of coordinates in which prover $j$ is special,
$T^{(j)} := \{i \in [n]: q_{i}^{(j)} = 1\}$.
In coordinate $i$, prover $j$ responds with $(T^{(j)}, i)$.

Let $a_1,\ldots,a_n$ be the sequence of special provers which the verifier picks.
We claim that if $(a_1,\ldots,a_n) \in S$ then the verifier accepts in 
all coordinates. Of course this happens with probability $\mu(S)$.

To see this, note first that the sets $T^{(1)},\ldots,T^{(r)}$
indeed form a partition of $[n]$ (because in each coordinate there
is exactly one special prover).
Next, $z_i^{(1)} = \ldots = z_i^{(r)} = i \in T^{(a_i)}$, 
by definition of $T^{(j)}$ and since prover $a_i$ is special in coordinate $i$. 

Finally, $\us \in S$, 
since for all $n$ coordinates $\us$ is exactly the string $(a_1,\ldots,a_n)$.
\end{proof}

\subsection{Connection to the universality proof in 
\texorpdfstring{\cite{FV02}}{[FV02]}}
\label{sec:fv02}

We explain the connection between our proof of 
Theorem~\ref{thm:dhj-game} and \cite{FV02}.
Recall our definition of a hypergraph homomorphism 
(Definitions~\ref{def:homomorphism} and~\ref{def:id}).

\begin{restatable}{definition}{qualifiedgood}
\label{def:qualified-good}
Let $\oQ$ be an $r$-prover question set and let
$\ouQ := \oQ^n$ be its $n$-fold parallel repetition.
Let $S \subseteq \ouQ$ with $\mu(S) = |S| / \left|\oQ\right|^n$
and let $\underline{f} = (f_1, \ldots, f_n)$ be a vector
of $n$ homomorphisms of $\oQ$. We say that
$\underline{f}$ is \emph{good} for $S$ if:
\begin{itemize}
  \item For every $\oq \in \oQ$ we have that 
    $\underline{f}(\oq) := (f_1(\oq), \ldots, f_n(\oq)) \in S$.
  \item There exists $i \in [n]$ such that $f_i$ is identity. 
\end{itemize}

We say that the question set $\oQ$ is $(n, \epsilon)$-good if
for every $S \subseteq \ouQ$ with $\mu(S) \ge \epsilon$ there exists
a vector of homomorphisms that is good for $S$. 
\end{restatable}

It turns out that if $\oQ$ is $(n, \epsilon)$-good, then
$\omega_{\oQ}(n) \le \epsilon$. This is the forbidden subgraph method
and it is presented in Section~\ref{sec:good}.
Verbitsky \cite{Ver95, FV02} discovered a related game construction: 
\begin{theorem}[\cite{FV02}]
\label{thm:fv-universal}
Let $\oQ$ be a connected, $r$-prover question set and 
$S \subseteq \ouQ$.
There exists an $r$-prover game $\cG_S$ with question set $\oQ$ such that:
\begin{itemize}
  \item If $\cG_S$ is trivial, then there exists a homomorphism vector 
    $\uf$ that is good for $S$.
  \item $\val(\cG_S^n) \ge \mu(S)$.
\end{itemize} 
\end{theorem}

Note that Theorem~\ref{thm:fv-universal} implies that the forbidden subgraph
method is universal in the sense that it gives the best possible bounds
on $\omega_{\oQ}(n)$:
\begin{corollary}[\cite{FV02}]
Let $\oQ$ be a connected, $r$-prover question set. Then,
\begin{align*}
  \omega_{\oQ}(n) =
\inf \left\{ \epsilon: \oQ \text{ is } (n,\epsilon)\text{-good} \right\} \; .
\end{align*}
\end{corollary}

There are a couple of caveats with regards to our formulation of 
Theorem~\ref{thm:fv-universal}:
First, we state it in terms of
homomorphisms instead of forbidden subgraphs as in \cite{FV02}. 
However, (with hindsight) both statements are easily seen to be equivalent. 
Second, the proof in \cite{FV02} is only for the two-prover case. However, it
generalizes to multiple provers in a natural way.

It turns out that our construction of the game $\cG_S$ for 
Theorem~\ref{thm:dhj-game} 
is an instantiation of the construction from Theorem~\ref{thm:fv-universal} for
the question set $\oQ_r$. Consequently, Theorem~\ref{thm:fv-universal}
has a shorter proof that assumes Theorem~\ref{thm:dhj-game}.
We find it instructive to present it below:

\begin{claim}
\label{cl:id-constant}
Let $r \ge 3$. The only homomorphisms of the question set $\oQ_r$ are
identity and constants. 
\end{claim}

\begin{proof}
As previously, we identify the edges in $\oQ_r$ with numbers $a \in [r]$.
Assume that there exists a homomorphism $f = (f^{(1)}, \ldots, f^{(r)})$
that maps edge $b$ to some $a \ne b$. We show that it must be $f = \1_a$.

First, since $b$ is mapped to $a$, we must have
$f^{(b)}(1) = 0$, $f^{(a)}(0) = 1$ and 
$f^{(j)}(0) = 0$ for every $j \notin \{a, b\}$.
$f^{(a)}(0) = 1$ implies that every edge $j \ne a$ is also
mapped to $a$ and, consequently, $f^{(j)}(0) = f^{(j)}(1) = 0$
for every $j \ne a$. But from this it follows that also edge $a$
must be mapped onto itself, hence $f = \1_a$.
\end{proof}

\begin{claim}
\label{cl:good-implies-cl}
Let $r \ge 3$ and $S \subseteq \ouQ_r \cong [r]^n$ such that 
there exists a homomorphism vector $\uf$ that is good for $S$.
Then, $S$ contains a combinatorial line. 
\end{claim}

\begin{proof}
Define a pattern $\ub$ as
\begin{align*}
b_i := \begin{cases}
j & \text{if $f_i = \1_j$,}\\
\ordstar & \text{if $f_i = \Id$.}
\end{cases}
\end{align*}
By Claim~\ref{cl:id-constant}, $\ub$ is well-defined. From the
second point in a definition
of a good vector, there is a star on at least one coordinate.
From the first point in that definition, $\ub(j) \in S$
for every $j \in [r]$. Consequently, $L(\ub) \subseteq S$.
\end{proof}

\begin{proof}[Proof of Theorem~\ref{thm:dhj-game}]
Let $r \ge 3$ and $S \subseteq [r]^n \cong \ouQ_r$ such that $S$ does not have
a combinatorial line. By Claim~\ref{cl:good-implies-cl},
there is no homomorphism vector good for $S$ (viewed as a subset of $\ouQ_r$).
But now the game from Theorem~\ref{thm:fv-universal} does the job. 
\end{proof}

\section{Lower Bounds on Multi-Prover Parallel Repetition}
\label{sec:lower-bounds}

In this section we explore some lower bounds on parallel repetition implied by
Theorem~\ref{thm:dhj-game}. 
Our main observation is Theorem~\ref{thm:multi-no-exp}:
for more than two provers
there exist question sets that do not admit exponential parallel
repetition.

\begin{proof}[Proof of Theorem~\ref{thm:multi-no-exp}]
Let $n$ be divisible by $r$ and let $S$ contain all strings with equidistributed
alphabet elements, i.e.,
\begin{align*}
S := \{ \ux \in [r]^n: w_1(\ux) = \ldots = w_r(\ux) = n/r \} \; .
\end{align*}
It is clear that $S$ does not contain a combinatorial line. At the same
time, by Stirling's approximation, $\mu(S) \ge \Omega(1 / n^{(r-1)/2})$
(where the constant in the $\Omega()$ notation depends on $r$) and therefore
$\omega_r^{\DHJ}(n)$ cannot decrease exponentially.

By Theorem \ref{thm:pr-dhj}, $\omega_{\oQ_r}(n)$ cannot decrease exponentially
either.
\end{proof}

Better lower bounds for $\omega_r^{\DHJ}(n)$ are known, with the best ones
established by the Polymath project \cite{Pol10}.

\begin{theorem}[\cite{Pol10}, Theorem 1.3]
\label{thm:dhj-lower-bound}
Let $\ell \ge 1$ and $r := 2^{\ell-1}+1$. There exists
$C_\ell > 0$ such that for every $n \ge 2$ there exists
a set $S \subseteq [r]^n$ with
\begin{align*}
\mu(S) \ge \exp \left(
- C_{\ell} \left(\log n\right)^{1/\ell} \right) \; 
\end{align*}
such that $S$ does not contain a combinatorial line.
\end{theorem}

That is, for $r = 2^{\ell-1} + 1$, we have
\begin{align} \label{eq:02c}
  \omega_{\oQ_r}(n) \ge \exp \left(
    - C_r \left( \log n \right)^{1 / \lceil \log r \rceil}
  \right) \; .
\end{align}

Inequality \eqref{eq:02c} is also interesting in the context of
the two-prover parallel repetition lower bound by Feige and Verbitsky
\cite{FV02}. Recall that the upper bound of  Raz (cf.~\eqref{eq:03c})
exhibits a dependence on the answer set size. More specifically,
it contains $1 / \log \left|\oA\right|$ term in the exponent. 
The example from \cite{FV02}
shows that if an exponential two-prover parallel repetition bound depends only 
on $\epsilon$ and $\left|\oA\right|$, this term cannot be larger than 
$\log \log \left|\oA\right| / \log \left|\oA\right|$.

Our example implies that we can bring down this last term to 
$\left(\log \log \left|\oA\right|\right)^{\epsilon}
/ \log \left|\oA\right|$ for any $\epsilon > 0$,
at the price of increasing the number of provers:

\begin{theorem}
Let $\ell \ge 2$, $r := 2^{\ell-1}+1$. There exists a constant $C_\ell > 0$
such that for each $n \ge 2$ there exists
an $r$-prover game $\cG$ with question set $\oQ_r$,
$\val(\cG) \le 1-1/r$ and an answer set $\oA$
with size $\left|\oA\right| \in [2^{rn}, 2^{2rn}]$
such that
\begin{align}
\label{eq:08a}
  \val(\cG^n) \ge \exp\left(-C_{\ell}n \cdot 
  \frac{\left( \log \log \left|\oA\right|  \right)^{1/\ell}}{\log \left|\oA\right|}\right) \; . 
\end{align}
\end{theorem}

\begin{proof}
Fix $\ell$ and $n$ and take the $r$-prover game $\cG_S$ from 
Theorem \ref{thm:dhj-game} for the set $S \subseteq [r]^n$ from 
Theorem \ref{thm:dhj-lower-bound}.
One verifies that $\cG_S$ has question set $\oQ_r$ and that
the answer alphabet
size is $\left|\oA\right| = (2^n \cdot n)^r \in [2^{rn}, 2^{2rn}]$.

Since $S$ has no combinatorial line, we have $\val(\cG_S) \le 1-1/r$
and 
$\val(\cG_S^n) \ge \mu(S) \ge 
\exp\left(-C_\ell \left(\log n\right)^{1/\ell}\right)$.

Noting that $n \ge \log\left|\oA\right| / 2r$ and
$\log n \le \log \log \left|\oA\right|$,
we can establish (\ref{eq:08a}):
\begin{align*}
\val(\cG_S^n) \
& \ge
\exp\left(-C_\ell \left(\log n\right)^{1/\ell}\right)
=
\exp \left(-C_{\ell}n \cdot \frac{\left(\log n\right)^{1/\ell}}{n} \right)
\\ & \ge
\exp \left(-C'_{\ell}n \cdot 
  \frac{\left(\log \log \left|\oA\right|\right)^{1/\ell}}
       {\log \left|\oA\right|} \right) \; .
\end{align*}
\end{proof}

As a final note, we reiterate that our lower bounds do \emph{not} exclude
the possibility of an ``information theoretic'' 
(see Section~\ref{sec:raz-compare})
parallel repetition bound. Furthermore,
all results of this section concern games with at least three provers.

\section{The Hales-Jewett Theorem and Coloring Games}
\label{sec:coloring}

As the name suggests, the density Hales-Jewett theorems is the density
version of the earlier Hales-Jewett theorem \cite{HJ63}. Inspired by
Theorem~\ref{thm:pr-dhj}, in this section we define \emph{coloring games}
and prove equivalence of their parallel repetition and the Hales-Jewett theorem.
We are not aware of any previous works concerning coloring games.  

\subsection{The Hales-Jewett theorem}
\begin{definition}
Let $r, c, n \in \bbN_{>0}$ and $C: [r]^n \to [c]$ be a coloring of $[r]^n$
with $c$ colors. We say that there is a \emph{monochromatic line}
in $C$ if there exists a combinatorial pattern $\ub$ such that
$C(\ub(1)) = \ldots = C(\ub(r))$.

For $r, n \in \bbN_{>0}$, let
\begin{align*}
\omega^{\HJ}_r(n) := \min \left\{
  c: \exists C: [r]^n \to [c] \text{ with no
  monochromatic lines}
\right\} \; .
\end{align*}
\end{definition}

\begin{theorem}[Hales-Jewett theorem]\label{thm:halesjewett}
For every $r \ge 2$:
\begin{align*}
\lim_{n \to \infty} \omega^{\HJ}_r(n) = +\infty \; .
\end{align*}
\end{theorem}

\begin{remark}
Even though the Hales-Jewett theorem follows easily from the density version,
better, primitive recursive bounds for $\omega^{\HJ}_r(n)$ are known
\cite{She88} compared to $\omega^{\DHJ}_r(n)$.
\end{remark}

\subsection{Coloring games}
\begin{definition}
An $r$-prover coloring game $\cG = (\oQ,\oA, V)$ is given by a 
question set $\oQ \subseteq Q^{(1)} \times \ldots \times Q^{(r)}$,
an answer set $\oA = A^{(1)} \times \ldots \times A^{(r)}$ and a function
$V: \oQ \times \oA \to \cX$ for some set $\cX$.

The \emph{color-value} of a game is 
\begin{align*}
\cVal(\cG) := \min_{\ocS = (\cS^{(1)},\ldots,\cS^{(r)})} 
\left|
V\left(\oQ, \ocS(\oQ)\right)
\right| \; .
\end{align*}
In other words, the provers are supposed to minimize the number of colors
that the verifier can output instead of maximizing the probability of 
acceptance.

Given a coloring game $\cG$, we define the parallel repetition similar
to before. 
The only change is that the function $V^n$ outputs a vector of $n$ values
which is obtained by applying $V$ to every coordinate individually.
\end{definition}

\begin{definition}
For an $r$-prover question set $\oQ$ let
\begin{align*}
\omega^{\CPR}_{\oQ}(n) := \min_{\cG} \cVal(\cG^n) \; ,
\end{align*}
where the minimum is over all coloring games $\cG$ with question set $\oQ$
and $\cVal(\cG) \ge 2$.
\end{definition}

\begin{theorem}\label{thm:coloringPR}
For every $r$-prover question set $\oQ$:
\begin{align*}
\lim_{n \to \infty} \omega^{\CPR}_{\oQ}(n) = +\infty \; .
\end{align*}
\end{theorem}

\subsection{Equivalence of the theorems}
\begin{proof}[Proof (Theorem~\ref{thm:halesjewett} implies 
Theorem~\ref{thm:coloringPR})]
In particular, we will show that
\begin{align}
\omega_{\oQ}^{\CPR}(n) \ge \omega_{r}^{\HJ}(n) \; ,
\end{align}
where $r = \left|\oQ\right|$.

Consider a coloring game $\cG$ with the question set $\oQ$
and suppose for a contradiction
that there is a strategy for $\cG^n$ 
where the verifier uses $c < \omega_r^{HJ}(n)$ colors.
Fix such a strategy and identify the colors used
in it with $[c]$.
Consider now the map $C: \oQ^n \to [c]$ which
tells us what color the verifier will output in the repeated
game for this strategy.

Theorem~\ref{thm:halesjewett}
implies that there is a pattern~$\ub$ such that $C(\ub(1)) = \ldots 
= C(\ub(r))$. 
This, however, implies that~$\cG$ has coloring value~$1$.
To see this, consider the following prover strategy:
on input $q^{(j)}$, prover $j$ applies the repeated strategy
with all $\star$ symbols in the pattern replaced with $q^{(j)}$,
and for each other position she computes the question on the pattern input.
Then, she responds with the response in the first star coordinate.
\end{proof}

\begin{proof}[Proof (Theorem~\ref{thm:coloringPR} implies
Theorem~\ref{thm:halesjewett})]
In case of doubts
the reader is advised to read the proof of Theorem~\ref{thm:dhj-game} first.
Recall the question set $\oQ_r$ from Definition~\ref{def:qr}.
In fact, we will prove that
\begin{align}
\omega^{\HJ}_r(n) \ge \omega^{\CPR}_{\oQ_r}(n) \; .
\end{align}

Let $c < \omega_{\oQ_r}^{\CPR}(n)$ and fix an arbitrary function 
$C: [r]^n \to [c]$. We want to show that $C$ has a monochromatic line.
To this end,
consider the following $r$-prover coloring game $\cG_C$ 
over $\oQ_r$: 

The answer alphabet is
$A^{(j)} := 2^{[n]} \times [n]$.
After receiving answers $(T^{(j)}, z^{(j)})$, the verifier does checks similar 
as in the proof of 
Theorem~\ref{thm:dhj-game}: The sets $T^{(j)}$ should partition $[n]$
inducing a string $\us \in [r]^n$ and there should be a special coordinate
$z = z^{(1)} = \ldots = z^{(r)}$ such that $z \in T^{(a)}$, where $a$ is
the special prover (i.e., the one that received 1).
 
If all the checks are passed, our verifier applies $C$ on $\us$
and outputs the resulting element of $[c]$.
Otherwise, the verifier outputs his question tuple $\oq$ 
(which we assume to be not in $[c]$).

There is a strategy for $\cG_C^n$ with coloring value $c$:
letting 
$T^{(j)}$ be the set of coordinates in which prover $j$ is special,
prover $j$ responds on coordinate $i$ with $(T^{(j)}, i)$.
Because $c < \omega_{\oQ_r}^{\CPR}(n)$, we have
$\cVal(\cG_C) = 1$. 

Consider a strategy for $\cG_C$ that uses only one color. Since there
are at least two questions, the verifier can never output $\oq$.
Consequently, $r \ge 3$ implies that there is a special coordinate
$z$ that the provers always output. Furthermore, the
sets $T^{(j)}(0)$ are pairwise disjoint and the special prover $a$
always outputs $T^{(a)}(0) \cup T(1)$, where
$T(1) := [n] \setminus \left( T^{(1)}(0) \cup \ldots \cup T^{(r)}(0) \right)$
with $z \in T(1)$. 

The coordinate sets $T^{(j)}(0)$ and $T(1)$ define a monochromatic
combinatorial line.
\end{proof}

\section{Constructability Implies Parallel Repetition}
\label{sec:constructability-pr}

In this section we first define a class of constructible hypergraphs
and then establish that all constructible question sets admit exponential
parallel repetition. The main result of this section is 
Theorem~\ref{thm:constructible-is-good-multi}.

\subsection{Constructing hypergraphs by conditioning}
\label{sec:constructing}

We define constructability in the general case, but for intuition the reader
is invited to think about bipartite graphs (i.e., $r=2$).
Recall our definitions of hypergraph homomorphisms
(Definitions~\ref{def:homomorphism} and~\ref{def:id}).

\begin{definition}
Given an $r$-regular, $r$-partite hypergraph
$(Q^{(1)}, \ldots, Q^{(r)}, \oQ)$
and sets $P^{(j)} \subseteq Q^{(j)}$
we define its \emph{section hypergraph}
$(P^{(1)}, \ldots, P^{(r)}, \oP)$,
where $\oP \subseteq \oQ$ consists of those
hyperedges whose vertices are all in
$P^{(1)} \cup \ldots \cup P^{(r)}$. 
\end{definition}

In the graph case the section hypergraph corresponds to an induced subgraph.

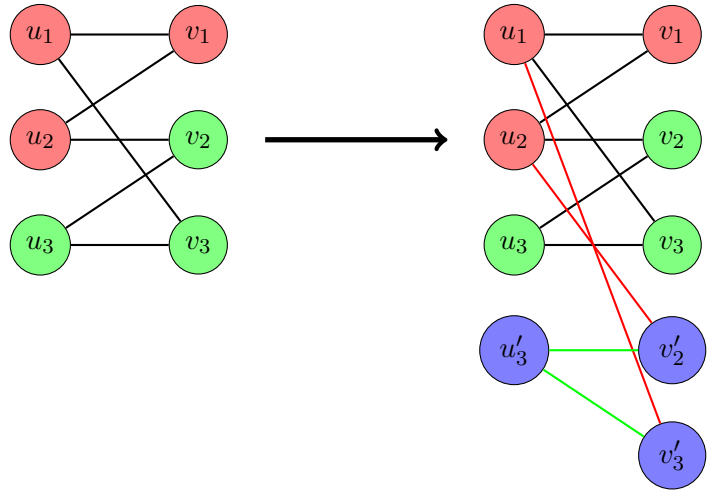
\begin{figure}\centering
\caption{Doubling a bipartite graph. Fixed vertices in red, old vertices in green,
new vertices in blue.}
\label{fig:doubling}
\begin{tikzpicture}[xscale=1.4, yscale=1.4]
\tikzset{every node/.style={circle, draw, text=black}}

\node at (0, 0.5) [draw=none] {};

\node (a0) at (0, 0) [fill=red!50] {$u_1$};
\node (a1)  at (0, -1) [fill=red!50] {$u_2$};
\node (a2)  at (0, -2) [fill=green!50] {$u_3$};
\node (b0)  at (1.5, 0) [fill=red!50] {$v_1$};
\node (b1)  at (1.5, -1) [fill=green!50] {$v_2$};
\node (b2) at (1.5, -2) [fill=green!50] {$v_3$};
\draw [thick] (a0) -- (b0);
\draw [thick] (a0) -- (b2);
\draw [thick] (a1) -- (b0);
\draw [thick] (a1) -- (b1);
\draw [thick] (a2) -- (b1);
\draw [thick] (a2) -- (b2);

\node (i1) at (2, -1) [draw=none]  {};
\node (i2) at (4, -1) [draw=none] {};
\draw [line width = 2, ->] (i1) -- (i2);

\node (c0) at (4.5, 0) [fill=red!50] {$u_1$};
\node (c1) at (4.5, -1) [fill=red!50] {$u_2$};
\node (c2) at (4.5, -2) [fill=green!50] {$u_3$};
\node (c3) at (4.5, -3) [fill=blue!50] {$u'_3$};
\node (d0) at (6, 0) [fill=red!50] {$v_1$};
\node (d1) at (6, -1) [fill=green!50] {$v_2$};
\node (d2) at (6, -2) [fill=green!50] {$v_3$};
\node (d3) at (6, -3) [fill=blue!50] {$v'_2$};
\node (d4) at (6, -4) [fill=blue!50] {$v'_3$};
\draw [thick] (c0) -- (d0);
\draw [thick] (c0) -- (d2);
\draw [thick] (c1) -- (d0);
\draw [thick] (c1) -- (d1);
\draw [thick] (c2) -- (d1);
\draw [thick] (c2) -- (d2);
\draw [color=red, thick] (c0) -- (d4);
\draw [color=red, thick] (c1) -- (d3);
\draw [color=green, thick] (c3) -- (d3);
\draw [color=green, thick] (c3) -- (d4);
\end{tikzpicture}
\end{figure}

\begin{definition}
\label{def:constructible-hypergraphs}
Let $r \ge 2$.
We recursively define the class of $r$-regular, 
$r$-partite hypergraphs that are 
\emph{constructible by conditioning}:
\begin{enumerate}
  \item A single hyperedge 
    $\left(\{q^{(1)}\}, \ldots, \{q^{(r)}\}, 
      \{(q^{(1)}, \ldots, q^{(r)})\}\right)$ 
    is constructible.
  \item If $(P^{(1)} \cupdot Q^{(1)}, \ldots,
    P^{(r)} \cupdot Q^{(r)}, \oP)$ is constructible, then
    $(P^{(1)} \cupdot Q^{(1)} \cupdot R^{(1)},
    \allowbreak \ldots, \allowbreak
    P^{(r)} \cupdot \allowbreak Q^{(r)} \cupdot R^{(r)},
    \allowbreak
    \oP \cup \oQ)$ is
    constructible, where:
    \begin{itemize}
    \item $R^{(j)} := \{ q': q \in Q^{(j)} \}$ is a set of copies of vertices 
      from $Q^{(j)}$.
      We say that the vertices in $P^{(j)}$ are
      \emph{fixed}, vertices from $Q^{(j)}$ are \emph{old}
      and vertices from $R^{(j)}$ are \emph{new}.
      \item
        We say that a hyperedge is fixed if all of its vertices are fixed.
        For a hyperdge $\oq \in \oP$ that is not fixed we define
        $\oq'$ as the hyperedge formed from $\oq$ by
        replacing all of its old vertices by their respective copies.

        Then, $\oQ := \left\{ \oq': \oq \in \oP, 
        \oq \text{ is not fixed} \right\}$.
    \end{itemize}
  In this case we say that $\oP \cup \oQ$ was constructed 
  from $\oP$ by \emph{doubling} $Q^{(1)} \cup \ldots \cup Q^{(r)}$.

  Figure \ref{fig:doubling} can be consulted for an example in the graph
  case.
  \item 
    If $(P^{(1)} \cupdot Q^{(1)}, \ldots, P^{(r)} \cupdot Q^{(r)}, \oQ)$
    is constructible and
    $(P^{(1)}, \ldots, P^{(r)}, \oP)$ is a section hypergraph of $\oQ$ 
    such that there exists a homomorphism from $\oQ$ to $\oP$ which is identity 
    on $P^{(1)} \cup \ldots \cup P^{(r)}$, then $\oP$ is constructible.

    In such case we say that \emph{$\oQ$ collapses onto $\oP$}.
\end{enumerate} 
\end{definition}

Observe that the doubling operation never produces hyperedges incident to
both old and new vertices.

To give some intuition on the conditioning operations we state two
simple properties.

\begin{claim}
Let $r \ge 2$. Every $r$-partite
hypergraph can be collapsed onto one of its hyperedges.
\end{claim}

\begin{definition}
The \emph{complete} hypergraph on $Q^{(1)}, \ldots, Q^{(r)}$
is $(Q^{(1)}, \ldots, Q^{(r)}, \allowbreak
Q^{(1)} \times \allowbreak \ldots \allowbreak \times Q^{(r)})$.
\end{definition}

\begin{claim}
\label{cl:full-constructible}
Let $r \ge 2$ and $Q^{(1)}, \ldots, Q^{(r)}$ be finite sets.
The complete hypergraph on $Q^{(1)}, \ldots, Q^{(r)}$ is constructible.
\end{claim}

\begin{proof}
Let $k^{(j)} := |Q^{(j)}|$, 
$Q^{(j)} = \{q^{(j)}_1, \ldots, q^{(j)}_{k^{(j)}}\}$
and $P^{(j)} := \{ q_1^{(j)} \}$. Start the construction with a single
hyperedge 
$(P^{(1)}, \ldots, P^{(r)}, P^{(1)} \times \ldots \times P^{(r)})$.

The complete hypergraph is constructed in $r$ stages. In the $j$-th
stage vertex $q^{(j)}_1$ is doubled $k^{(j)}-1$ times with all other vertices
fixed. Observe that after the $j$-th stage the current hypergraph is
the complete hypergraph on $Q^{(1)}, \ldots, Q^{(j)}, P^{(j+1)}, \ldots, P^{(r)}$.
\end{proof}

\begin{remark}
\label{rem:full-constructible}
As a matter of fact, if $|Q^{(1)}| = \ldots = |Q^{(r)}| = 2^k$, then it
is not difficult to see that the complete hypergraph on 
$Q^{(1)}, \ldots, Q^{(r)}$ can be constructed with 
$rk$ doublings.
\end{remark}

\subsection{Constructability implies parallel repetition}
\label{sec:constructability-pr-proof}

Our goal in this section is the following quantitative version
of Theorem~\ref{thm:constructible-good-simple}:

\begin{theorem}
\label{thm:constructible-is-good-multi}
Let $\oQ$ be an $r$-prover question set that 
is constructible by conditioning using $k$ doublings 
(and an arbitrary number of collapses).
Let $M := \left| \oQ \right|$. Then,
\begin{align*}
\omega_{\oQ}(n) \le 3\exp\left(-n / M^{2^{k+1}}\right) \; .
\end{align*}
In particular, $\oQ$ admits exponential parallel repetition.
\end{theorem}

A very rough proof outline is as follows: first, exponential parallel 
repetition is equivalent to exponential decrease of the threshold for good 
homomorphism vectors (cf.~Definition~\ref{def:qualified-good}). Second, we show that 
the existence of good homomorphism vectors is implied by existence of
a probability distribution over $\Hom\left(\oQ, \oQ\right)$ with certain set hitting
properties. Third, we prove that such distribution exists for every
constructible $\oQ$.

We elaborate those three steps in the next subsections.

\subsubsection{Good question sets}
\label{sec:good}

For convenience we restate Definition~\ref{def:qualified-good}:

\qualifiedgood*

Observe that a vector of $n$ identities $\uf = (\Id, \ldots, \Id)$
is good for the whole space $\ouQ$
and hence every question set $\oQ$ is $(n, 1)$-good.

\begin{remark}
Note that if $\oQ$ is $(n, \epsilon)$-good, then it is also 
$(n+1, \epsilon)$-good. This is because given $S \subseteq \oQ^{n+1}$ we can 
set $f_{n+1}$ to be a constant homomorphism such that the (relative) measure 
$\mu(S)$ does not decrease
conditioned on $\oq_{n+1} = f_{n+1}(\oq)$. Then we can get $f_1, \ldots, f_n$ 
from the assumption that $\oQ$ is $(n, \epsilon)$-good.
\end{remark}

\begin{definition}
Let $\oQ$ be a question set and $n \in \mathbb{N}_{>0}$.
We define
\begin{align*}
\omega_{\oQ}^{\good}(n) := \inf \left\{
\epsilon: \text{$\oQ$ is $(n, \epsilon)$-good}
\right\} \; .
\end{align*}
We say that $\oQ$ is \emph{good} if 
$\lim_{n \to \infty} \omega_{\oQ}^{\good}(n) = 0$.
\end{definition}

The value of $\omega^{\good}_{\oQ}(n)$ is an upper bound on the parallel repetition
rate $\omega_{\oQ}(n)$:

\begin{lemma}
\label{lem:good-implies-pr}
Let $\oQ$ be a question set. Then,
\begin{align*}
\omega_{\oQ}(n) \le \omega^{\good}_{\oQ}(n) \; .
\end{align*}
\end{lemma}

\begin{proof}
Assume otherwise, i.e., that there exists a game $\cG$ with question 
set $\oQ$ and $\val(\cG) < 1$ such that 
$\val(\cG^n) > \omega_{\oQ}^{\good}(n)$. We construct a perfect strategy
for $\cG$, which is a contradiction.

Fix an optimal strategy for $\cG^n$ and let
$S \subseteq \ouQ$ with $\mu(S) > \omega_{\oQ}^{\good}(n)$ 
be the set of question vectors in the repeated game
for which the players win.

Let $\underline{f}$ be a vector of homomorphisms of $\oQ$ that is good
for $S$ and let $i$ be a coordinate where $f_i$ is identity.

A strategy for the game $\cG$ for the $j$-th prover is as follows:
Given $q^{(j)} \in Q^{(j)}$, obtain 
$\underline{f}^{(j)}(q^{(j)}) = (f_1^{(j)}(q^{(j)}), \ldots, f_n^{(j)}(q^{(j)}))$.
Then, consider the answer of the $j$-th prover on $\underline{f}^{(j)}(q^{(j)})$
in the strategy for $\cG^n$. Finally, output the $i$-th coordinate
of that answer.

Since for every $\oq = (q^{(1)}, \ldots, q^{(r)}) \in \oQ$ we have that
$\underline{f}(\oq) = (\underline{f}^{(1)}(q^{(1)}), \allowbreak\ldots,
\allowbreak \underline{f}^{(r)}(q^{(r)}))
\in S$, when applying the above strategy the provers
are always winning on all coordinates of $\cG^n$.
Since $f_i(\oq) = \oq$, their answers on 
the $i$-th coordinate are winning for $\oq$ in the game $\cG$. Therefore,
$\val(\cG) = 1$, a contradiction.
\end{proof}

\begin{remark}
Lemma~\ref{lem:good-implies-pr} is essentially what is called in the literature 
the forbidden subgraph method. However, our formulation with homomorphisms
is different than the one that uses forbidden subgraphs, e.g., in \cite{FV02}.

Verbitsky \cite{Ver95, FV02} showed that the forbidden subgraph method
is universal, i.e., for a connected question set $\oQ$ there is
$\omega_{\oQ}(n) = \omega_{\oQ}^{\good}(n)$
(cf.~Theorem~\ref{thm:fv-universal}).
\end{remark}

\subsubsection{Proving that 
\texorpdfstring{$\oQ$}{Q}
is good with probabilistic method}

\begin{lemma}
\label{lem:probabilistic-good}
Let $\oQ$ be a question set
and let $\mathcal{H}$ be a distribution
over $\Hom(\oQ, \oQ)$ such that:
\begin{enumerate}
\item 
If $\underline{f} = (f_1, \ldots, f_n)$ is sampled such that
$f_i$ is i.i.d.~in $\cH$, then:
\begin{align*}
\forall S \subseteq \ouQ: \Pr \left[
  \forall \oq \in \oQ: \underline{f}(\oq) \in S
\right] \ge c(\mu(S)) \; ,
\end{align*}
where $c(\mu) > 0$ if $\mu > 0$.
\item $\cH(\Id) > 0$.
\end{enumerate}
Then, $\oQ$ is good. Furthermore, if $\cH(\Id) \ge \epsilon > 0$
and $c(\mu) \ge \mu^C / C$ for some $C \ge 1$, then
$\omega_{\oQ}^{\good}(n) \le 3\exp(-\epsilon n / C)$.
\end{lemma}

\begin{proof}
  Let $\epsilon := \cH(\Id)$ and $\mu \in (0,1]$. For $S \subseteq \ouQ$
  with $\mu(S) = \mu$, define
  the event
  \begin{align*}
  \cE :\equiv \forall \oq \in \oQ: \underline{f}(\oq) \in S \land 
    \exists i \in [n]: f_i = \Id \; .
  \end{align*}
  Since $\Pr[\forall \oq \in \oQ: \underline{f}(\oq) \in S] \ge c(\mu)$ and
  $\Pr[\exists i: f_i = \Id] = 1-(1-\epsilon)^n$, by union bound, if
\begin{align}
\label{eq:04a}
\Pr[\forall i: f_i \ne \Id] = (1-\epsilon)^n \le c(\mu) / 2 \; ,
\end{align}
then $\Pr[\cE] > 0$. Therefore, if we choose $n$ such that (\ref{eq:04a})
holds, then $\oQ$ is $(n, \mu)$-good. Since for arbitrary 
$\mu \in (0, 1]$ we found
that $\oQ$ is $(n, \mu)$-good for $n$ big enough, $\oQ$ must be good.

Furthermore, if $c(\mu) \ge \mu^C / C$, setting
\begin{align*}
\mu := \left(2C(1-\epsilon)^n\right)^{1/C} \le 
(2C)^{1/C} \cdot \exp(-\epsilon n / C) \le 3 \exp (-\epsilon n / C) \; ,
\end{align*}
we see that:
\begin{align*}
(1-\epsilon)^n = \mu^C / 2C \le c(\mu) / 2 \; ,
\end{align*}
and therefore $\omega_{\oQ}^{\good}(n) \le 3\exp(-\epsilon n / C)$.
\end{proof}

\subsubsection{Same-set hitting homomorphism spaces}

\begin{lemma}
\label{lem:constructible-is-good-induction}
Let $\oP$ be an $r$-partite hypergraph
constructible using $k$ doublings (and an arbitrary number of
collapses) and let $\oQ$ be another $r$-partite hypergraph.

Then, there exists a distribution $\cH$ over $\Hom(\oP, \oQ)$
such that:
\begin{enumerate}
\item
If $\underline{f} = (f_1, \ldots, f_n)$ is sampled such that $f_i$ is
i.i.d.~in $\cH$, then:
\begin{align*}
  \forall S \subseteq \ouQ: \Pr \left[
    \forall \op \in \oP: \underline{f}(\op) \in S
  \right] \ge \mu(S)^{C} \; ,
\end{align*}
where $C = 2^k$.
\item
$\min_{f \in \Hom(\oP, \oQ)} \cH(f) \ge 1 / M^C$,
where $C = 2^k$ and $M = \left|\oQ\right|$.
\end{enumerate}
\end{lemma}

This lemma is inspired by the paper \cite{HHM15} in the following way:
Let $f$ be a random homomorphism sampled according to $\cH$ and
let $\oP = \{\op^{(1)}, \ldots, \op^{(k)}\}$. We can think of $\cH$
as a $k$-step random process with the steps given by
$f(\op^{(1)}), \ldots, f(\op^{(k)})$.
Then, the first condition in Lemma \ref{lem:constructible-is-good-induction}
is equivalent to saying that $\cH$ is polynomially same-set hitting
as defined in \cite{HHM15}.

Later we will apply Lemma \ref{lem:constructible-is-good-induction} with
$\oP = \oQ$.

\begin{proof}
The proof proceeds by induction on the structure of $\oP$.
To achieve the constant $C$ as claimed, we need to show the base case with 
$C = 1$ and then argue that a collapse preserves $C$ and that a doubling 
increases $C$ at most twice.

\begin{enumerate}
\item
If $\oP$ is a single hyperedge, then $\Hom(\oP, \oQ)$ is isomorphic to $\oQ$. 
Setting $\cH(f_{\oq}) := 1 / M$ for $\oq \in \oQ$ one can
easily see that both $1$ and $2$ are satisfied with $C = 1$.
\item
Assume that $\oP$ was constructed by doubling a hypergraph $\oP_0$.
Let $A$ be the set of fixed vertices, $B$ the old vertices and $B'$ the new vertices
(regardless of the player they belong to).
Therefore the vertex set of $\oP_0$ is $A \cup B$ and the vertex set of $\oP$
is $A \cup B \cup B'$.

We are going to write homomorphisms $f \in \Hom(\oP_0, \oQ)$ as 
$f = (f_A, f_B)$ and $f \in \Hom(\oP, \oQ)$ as $f = (f_A, f_B, f_{B'})$.

Observe that
\begin{align}
  \Hom(\oP, \oQ) =& \big\{ 
  (f_A, f_B, f_{B'}): (f_A, f_B) \in \Hom(\oP_0, \oQ) 
  \nonumber 
  \\ \label{eq:05a} &              
  \quad \land (f_A, f_{B'}) \in \Hom(\oP_0, \oQ)
  \big\} \; ,
\end{align}
where we abused the notation in the expression $(f_A, f_{B'})$:
this is justified from the definition of the doubling operation.

By induction, there exists a distribution $\cH_0$ on $\Hom(\oP_0, \oQ)$ satisfying
$1$ and $2$ for some $C_0 > 0$. Let $H = (H_A, H_B)$ be a random variable
distributed according to $\cH_0$. Define:
\begin{IEEEeqnarray}{rCl}
  \cH(f_A, f_B, f_{B'}) &:=&
\Pr[ H_A = f_A ] \cdot
  \Pr[ H_B = f_B \mid H_A = f_A] \nonumber \\
  \label{eq:06a} &&
  \cdot \Pr [ H_B = f_{B'} \mid H_A = f_A ] \; .
\end{IEEEeqnarray}

By (\ref{eq:05a}), (\ref{eq:06a}) defines a probability 
distribution. Furthermore:
\begin{align*}
\cH(f_A, f_B, f_{B'}) \ge \cH_0(f_A, f_B) \cdot \cH_0(f_A, f_{B'}) \ge
\epsilon^{2C_0} \; .
\end{align*}

As for condition $1$, let $\oE_A$ be the fixed hyperedges of $\oP_0$
(i.e., those that have all their vertices in $A$)
and $\oE_B$ and $\oE_{B'}$ be the hyperedges of $\oP$ that have vertices
incident to $B$ and $B'$, respectively. Note that $\oE_A$ and $\oE_B$ form
a partition of $\oP_0$ and $\oE_A$, $\oE_B$ and $\oE_{B'}$ form a partition
of $\oP$. 

Recall that $\underline{f} = (f_1, \ldots, f_n)$ is a random vector
with coordinates sampled i.i.d.~from $\cH$.
We are going to decompose 
$\underline{f} = (\underline{f}_A, \underline{f}_B, \underline{f}_{B'})$
in the natural way.
Fix $S \subseteq \oQ^n$ and 
define the event $\cE :\equiv \forall \op \in \oE_A: \underline{f}(\op) \in S$.

We estimate, using Jensen's inequality in \eqref{eq:07a}:
\begin{IEEEeqnarray*}{rCl}
  \IEEEeqnarraymulticol{3}{l}{
  \Pr \left[ \forall \op \in \oP: \underline{f}(\op) \in S \right]
  }\\ & = & 
  \EE \left[
    \EE \left[
      \mathbbm{1}_\cE \cdot \Pr \left[
        \forall \op \in \oE_B \cup \oE_{B'}: \underline{f}(\op) \in S
      \mid \underline{f}_A \right] 
    \mid \underline{f}_A \right]
  \right]
  \\ & = &
  \EE \left[
    \EE \left[
      \mathbbm{1}_\cE \cdot \Pr \left[
        \forall \op \in \oE_B: \underline{f}(\op) \in S
      \mid \underline{f}_A \right]^2 
    \mid \underline{f}_A \right]
  \right]
  \\ & = & 
  \EE \left[
    \EE \left[
      \left( \mathbbm{1}_\cE \cdot \Pr \left[
        \forall \op \in \oE_B: \underline{f}(\op) \in S
      \mid \underline{f}_A \right] \right)^2 
    \mid \underline{f}_A \right]
  \right]
  \\ & = &
  \EE \left[
    \EE \left[
      \mathbbm{1}_\cE \cdot \Pr \left[
        \forall \op \in \oE_B: \underline{f}(\op) \in S
      \mid \underline{f}_A \right] 
    \mid \underline{f}_A \right]^2
  \right]
  \\ & \ge &
  \EE \left[
    \EE \left[
      \mathbbm{1}_\cE \cdot \Pr \left[
        \forall \op \in \oE_B: \underline{f}(\op) \in S
      \mid \underline{f}_A \right] 
    \mid \underline{f}_A \right]
  \right]^2
  \IEEEyesnumber \label{eq:07a}
  \\ & = &
  \Pr \left[
    \forall p \in P_0: \underline{f}(p) \in S
  \right]^2
  \ge
  \mu^{2C_0} \; .
\end{IEEEeqnarray*}
\item
The last case considers $\oP$ constructed by collapsing
some $\oP_0$. Let $A$ be the vertex set of $\oP$ and $A \cupdot B$ the vertex
set of $\oP_0$.
Let $h \in \Hom(\oP_0, \oP)$ be a homomorphism that defines this collapse.

By induction, there exists a distribution $\cH_0$ on $\Hom(\oP_0, \oQ)$ 
satisfying properties $1$ and $2$ for some $C_0$. For $f \in \Hom(\oP, \oQ)$, 
define
\begin{align*}
\cH(f) := \sum_{\substack{g \in \Hom(\oP_0, \oQ)\\g_A = f}} \cH_0(g) \; .
\end{align*}

Since a restriction of a homomorphism is a homomorphism, $\cH$ indeed is a
probability distribution. Furthermore, since
$\cH(f) \ge \cH(h \circ f) \ge \epsilon^{C_0}$, condition $2$ is satisfied.

Finally, let $\underline{f}_0$ be a vector of question
homomorphisms sampled i.i.d.~from $\cH_0$
and recall that vector $\underline{f}$ is sampled
i.i.d.~from $\cH$. To establish condition 1, we check that
\begin{align*}
\Pr [ \forall \op \in \oP: \underline{f}(\op) \in S ]
& = \Pr [ \forall \op \in \oP: \underline{f}_0(\op) \in S]
\\ & \ge \Pr [ \forall \op \in \oP_0: \underline{f}_0(\op) \in S] \ge \mu^{C_0} \; .
\end{align*}
\end{enumerate}
\end{proof}

\subsubsection{Putting things together}

\begin{proof}[Proof of Theorem~\ref{thm:constructible-is-good-multi}]
Let $\oQ$ be an $r$-prover question set constructible by conditioning
using $k$ doublings with $\left|\oQ\right| = M$.

By Lemma \ref{lem:constructible-is-good-induction} applied for $\oP = \oQ$,
Lemma \ref{lem:probabilistic-good} and Lemma \ref{lem:good-implies-pr},
\begin{align*}
\omega_{\oQ}(n) \le
\omega_{\oQ}^{\good}(n) \le
3\exp\left( -n / 2^kM^{2^k}\right) \le
3\exp\left( -n / M^{2^{k+1}}\right) \; .
\end{align*}
Since $3\exp(-\alpha n) \le \exp(-\alpha n/2)$ for $n$ big enough, this
implies that $\oQ$ admits exponential parallel repetition.
\end{proof}

In particular, this recovers exponential parallel repetition
for free games:

\begin{corollary}
\label{cor:free-pr}
Let $\cG$ be a free $r$-prover game with $2^k$ questions available to each
prover. If $\cG$ is non-trivial, then
\begin{align*}
\val(\cG^n) \le 3\exp(-n / M^{2M}) \; .
\end{align*}
\end{corollary}

\begin{proof}
By Remark \ref{rem:full-constructible}, the question set of game
$\cG$ can be constructed using $rk$ doublings.
The bound then follows from Theorem \ref{thm:constructible-is-good-multi}:
\begin{align*}
  \val(\cG^n) \le \omega_{\oQ}(n) \le 3\exp\left( -n / M^{2^{rk + 1}} \right)
  = 3\exp\left( -n / M^{2M}\right) \; .
\end{align*}   
\end{proof}

We note that quantitatively the bound for free games from Corollary 
\ref{cor:free-pr} is much worse than the best known one by Feige \cite{Fei91},
which is $\exp\left(-\Omega\left(n / M \log M\right) \right)$.

\section{Constructing Graphs with Treewidth Two}
\label{sec:treewidth-two}

We turn to presenting the power of our system for proving parallel repetition.
In particular, we show that all two-prover graphs with treewidth at most two
are constructible.

Since in this section we deal only with two provers, we use more standard
notation where a bipartite graph is denoted as $G = (X, Y, E)$.
We will sometimes refer to vertices from $X$ as ``on the left'' and from
$Y$ as ``on the right''.

Our main result here is Theorem \ref{thm:tw-constructible}.

\subsection{Warm-up: forests are constructible}

We start with showing that all forests are constructible, recovering the
parallel repetition result by Verbitsky \cite{Ver95}. We will later use  
Lemma~\ref{lem:construct-leaf} in the construction of series-parallel graphs.

Firstly, we note that it is only interesting to consider constructability
of connected graphs (note that to create a new connected component one
can double all vertices of an existing connected component):
\begin{claim}
\label{cl:construct-connected}
A bipartite graph $G$ is constructible by conditioning if and only if all
its connected components are constructible.
\end{claim}

We can always add a ``fresh'' leaf to a constructible graph $G$:
\begin{lemma}
\label{lem:construct-leaf}
If $G = (X \cupdot \{u\}, Y, E)$ is constructible, then
$G' = (X \cupdot \{u\}, \allowbreak Y \cupdot \{v\}, \allowbreak
E \cup \{(u, v)\})$ is also constructible.
\end{lemma}

\begin{proof}
Pick an arbitrary edge $(u, w)$ originating from $u$.
Fix $u$ and double all the other vertices. Then collapse all new vertices
on the left onto $u$ and all new vertices on the right onto $w'$ 
(i.e., the copy of $w$).
\end{proof}

From Claim \ref{cl:construct-connected}, iterated application of Lemma
\ref{lem:construct-leaf} and Theorem \ref{thm:constructible-is-good-multi}
we have:
\begin{theorem}
Let $G$ be a tree. Then, $G$ is constructible by conditioning.
In particular, if $G$ is interpreted as a two-prover question set,
then it admits exponential parallel repetition.
\end{theorem}

\subsection{Treewidth and series-parallel graphs}

\begin{definition}[Treewidth]
Let $G$ be a simple graph. A \emph{tree decomposition} of $G$ is a tree
$T$, where each node (also called a \emph{bucket}) corresponds to a subset
of the vertices of $G$, with the following properties:
\begin{itemize}
  \item For each vertex $v$ of $G$, the buckets in which $v$ appears
    form a non-empty, connected subgraph of $T$.
  \item For each edge $e$ of $G$, there exists a bucket that contains both
    endpoints of $e$.
\end{itemize}
The \emph{width} of a tree decomposition of $G$ is the size of the biggest
bucket minus one.
The \emph{treewidth} of $G$ denoted by $\tw(G)$ is the smallest possible width
of a tree decomposition of $G$.
\end{definition}

We will not discuss treewidth here, referring the reader to any standard 
textbook on graph theory. We note that a connected graph has treewidth one if 
and only if it is a tree.

To characterise graphs with treewidth two, we need to introduce the notion
of \emph{generalized series-parallel graphs}.

\begin{definition}[Series-parallel graphs]
\label{def:series-parallel}
Let $G = (X, Y, E)$ be a bipartite graph and $u, v \in X \cup Y$.
We call a tuple $(X, Y, E, u, v)$ an \emph{oriented bipartite graph}. We call 
the  vertex $u$ the \emph{top} and $v$ the \emph{bottom}. 

We define the class of generalized bipartite series-parallel oriented 
(in short: series-parallel oriented) graphs
recursively:
\begin{enumerate}
\item Let $G = (\{a\}, \{b\}, \{(a, b)\})$ be a single edge. Then, both
$(G, a, b)$ and $(G, b, a)$ are series-parallel oriented graphs.

\item Let $G_1 = (X_1, Y_1, E_1, u, v)$ and
$G_2 = (X_2, Y_2, E_2, v, w)$
be series-parallel oriented graphs such that
$(X_1 \cup Y_1) \cap (X_2 \cup Y_2) = \{v\}$
and $v \in (X_1 \cap X_2) \cup (Y_1 \cap Y_2)$.

Then, $G := (X_1 \cup X_2, Y_1 \cup Y_2,  E_1 \cup E_2, u, w)$
is a series-parallel oriented graph.

We say that $G$ is a \emph{series composition} of $G_1$ and $G_2$
with $G_1$ on top and $G_2$ at the bottom.

\item Let $G_1 = (X_1, Y_1, E_1, u, v)$ and $G_2 = (X_2, Y_2, E_2, v, w)$
be series-parallel oriented graphs satisfying the same preconditions as
for the series composition.

Then, both $G := (X_1 \cup X_2, Y_1 \cup Y_2, E_1 \cup E_2, u, v)$ and
$G' := (X_1 \cup X_2, Y_1 \cup Y_2, E_1 \cup E_2, v, w)$ are series-parallel
graphs.

We say that $G$ and $G'$ are a \emph{generalized series composition} of $G_1$
and $G_2$. We say that $G_1$ is the \emph{primary graph} of $G$ and that
$G_2$ is the primary graph of $G'$.

\item Let $G_1 = (X_1, Y_1, E_1, u, v)$ and $G_2 = (X_2, Y_2, E_2, u, v)$ be 
series-parallel oriented graphs such that
$(X_1 \cup Y_1) \cap (X_2 \cup Y_2) = \{u, v\}$ and
$\{u, v\} \subseteq (X_1 \cap X_2) \cup (Y_1 \cap Y_2)$.

Then, $G := (X_1 \cup X_2, Y_1 \cup Y_2, E_1 \cup E_2, u, v)$ is also a
series-parallel oriented graph.

We call $G$ a \emph{parallel composition} of $G_1$ and $G_2$.
\end{enumerate}
We say that a bipartite graph $G$ is series-parallel if there exist vertices
$u, v$ such that $(G, u, v)$ is an oriented series-parallel graph.
\end{definition}

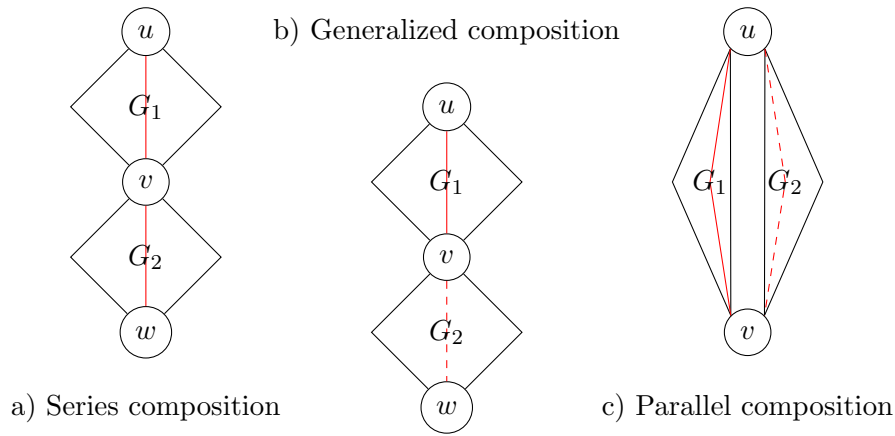
\begin{figure}[ht]\centering
\caption{Illustration of the composition operations.
The spines are drawn in continuous red. 
Note that $\cS(G_2)$ is not part of $\cS(G)$ in cases
of generalized and parallel composition,
hence the red dashed line.}
\label{fig:series-parallel}
\begin{tikzpicture}
\tikzset{every node/.style={circle, text=black}}
\node at (0, 0.5) {};

\node (ul) at (0, 0) [draw] {$u$};
\node (vl) at (0, -2) [draw] {$v$};
\node (wl) at (0, -4) [draw] {$w$};

\draw (ul) -- (-1, -1) -- (vl) -- (-1, -3) -- (wl);
\draw (ul) -- (1, -1) -- (vl) -- (1, -3) -- (wl);

\draw [red] (ul) -- (vl) -- (wl);

\node at (0, -1) {$G_1$};
\node at (0, -3) {$G_2$};
\node at (0, -5) {a) Series composition};

\node (uc) at (4, -1) [draw] {$u$};
\node (vc) at (4, -3) [draw] {$v$};
\node (wc) at (4, -5) [draw] {$w$};

\draw (uc) -- (3, -2) -- (vc) -- (3, -4) -- (wc);
\draw (uc) -- (5, -2) -- (vc) -- (5, -4) -- (wc);

\draw [red] (uc) -- (vc);
\draw [red, dashed] (vc) -- (wc);

\node at (4, -2) {$G_1$};
\node at (4, -4) {$G_2$};
\node at (4, 0) {b) Generalized composition};

\node (ur) at (8, 0) [draw] {$u$};
\node (vr) at (8, -4) [draw] {$v$};

\draw (ur.south west) -- (vr.north west);
\draw (ur.south east) to (vr.north east);
\draw (ur.south west) -- (7, -2) -- (vr.north west);
\draw (ur.south east) -- (9, -2) -- (vr.north east);

\draw [red] (ur.south west) -- (7.5, -2) -- (vr.north west);
\draw [red, dashed] (ur.south east) -- (8.5, -2) -- (vr.north east);

\node at (7.5, -2) {$G_1$};
\node at (8.5, -2) {$G_2$};
\node at (8, -5) {c) Parallel composition};
\end{tikzpicture}
\end{figure}

We refer the reader to Figure \ref{fig:series-parallel} for intuitive 
understanding of the composition operations. 

The requirement that the vertices by which the bipartite graphs are joined 
belong to the set $(X_1 \cap X_2) \cup (Y_1 \cap Y_2)$ ensures that they
belong to the same side of the graph and therefore the bipartedness is
preserved.
On the other hand, observe that the top and the bottom can lie either on the
same or the opposite sides of the bipartite graph.

In the literature the (not necessarily bipartite) graphs constructed with 
series and parallel composition are usually called series-parallel, and 
graphs that are constructed also with generalized composition are
called generalized series-parallel. Incidentally, a connected
graph is generalized series-parallel if and only if all its biconnected
components are series-parallel (see, e.g., \cite{Bod07}).

From now on, by ``series-parallel'' we will always mean the generalized
bipartite series-parallel graph from Definition \ref{def:series-parallel}.

We will use the following useful characterisation of graphs with treewidth
at most two:

\begin{theorem}
\label{thm:tw-vs-sp}
A connected bipartite graph $G$ has treewidth at most two if and only if $G$
is series-parallel.
\end{theorem}

For a proof of Theorem \ref{thm:tw-vs-sp} see \cite{HHC99}. Their
proof concerns the case of general (non-bipartite) graphs, but it is easy
to see that a connected generalized series-parallel graph is bipartite if 
and only if it can be constructed with additional restrictions as in 
Definition \ref{def:series-parallel}.

\subsection{Generalized series-parallel construction}

Recall that the main theorem of this section is:
\twtwo*

Due to Theorem \ref{thm:constructible-is-good-multi}, Claim
\ref{cl:construct-connected} and Theorem \ref{thm:tw-vs-sp}, to establish
Theorem \ref{thm:tw-constructible} it is enough to show that series-parallel
graphs are constructible. We spend the rest of this section to achieve
that goal.

\begin{definition}
Let $G$ be an oriented series-parallel graph. We define its (not oriented)
subgraph $\cS(G)$ and call it its \emph{spine}. The definition follows the recursive 
pattern of Definition \ref{def:series-parallel}:
\begin{enumerate}
\item If $G$ is a single edge, its spine is the whole of $G$.
\item If $G$ is a series composition of $G_1$ and $G_2$,
then $\cS(G)$ consists of $\cS(G_1)$ and $\cS(G_2)$ taken together.
\item If $G$ is a generalized composition of $G_1$ and $G_2$ with $G_1$
as the primary graph, then $\cS(G)$ is equal to $\cS(G_1)$.
\item If $G$ is a parallel composition of $G_1$ and $G_2$ and $\cS(G_1)$ 
has no more edges than $\cS(G_2)$, then $\cS(G)$
is equal to $\cS(G_1)$. Otherwise, it is equal to $\cS(G_2)$.
\end{enumerate}
\end{definition}

Observe that the spine is always an induced path between
the top and the bottom of $G$. As a matter of fact, it is a shortest path
from top to bottom in $G$. Furthermore,
the length of the spine $L(G)$ is given as:
\begin{enumerate}
\item One, if $G$ is a single edge.
\item $L(G_1)+L(G_2)$, if $G$ is a series composition of $G_1$ and $G_2$.
\item $L(G_1)$, if $G$ is a generalized composition of $G_1$ and $G_2$ with
$G_1$ as the primary graph.
\item $\min(L(G_1), L(G_2))$, if $G$ is a parallel composition of $G_1$ and
$G_2$.
\end{enumerate}
Finally, note that if $G$ is a parallel composition of $G_1$ and $G_2$, then
due to the bipartedness $L(G_1)$ and $L(G_2)$ must have the same parity.

Recall the graph construction operations from Definition 
\ref{def:constructible-hypergraphs}. A series-parallel graph can always be
collapsed onto its spine:

\begin{lemma}
\label{lem:spine-collapse}
Let $G$ be an oriented series-parallel graph. Then, $G$ (treated as an
unoriented graph) can be collapsed onto its spine.
\end{lemma}

\begin{proof}
By induction on the series-parallel structure of $G$. 
If $G$ is a single edge, it is clear.
If $G$ is a series  composition of $G_1$ and $G_2$, then by induction 
$G_1$ and $G_2$ can be collapsed onto their respective spines.

If $G$ is a generalized composition of $G_1$ and $G_2$, assume w.l.o.g.~that 
$G_1$ is the primary
graph and let $v$ be the bottom vertex of $G_1$. Then, by induction, $G_1$ can
be collapsed onto its spine $\cS(G_1) = \cS(G)$. 
On the other hand, all of $G_2$ can be collapsed
onto the edge $(v, w)$, where $w$ is the neighbor of $v$ in $\cS(G_1)$.

In case $(G, u, v)$ is a parallel composition of $G_1$ and $G_2$, assume w.l.o.g.~that
$\cS(G_1)$ is not longer than $\cS(G_2)$. Firstly, observe that
the spine $\cS(G_2)$ can be collapsed onto $\cS(G_1)=\cS(G)$: indeed, if
we write the vertices of $\cS(G_1)$ top-bottom as 
$(u_0 = u, u_1, \ldots, u_k = v)$ and analogously $\cS(G_2)$ as
$(v_0 = u, v_1, \ldots, v_{k+2\ell} = v)$, then the mapping:
\begin{align*}
  f(u_i) &:= u_i\\
  f(v_i) &:= \begin{cases}
    u_i & \text{if $i \le k$,}\\
    u_{k-(j \bmod 2)} & \text{if $i = k+j$,}
  \end{cases}
\end{align*}
is a required homomorphism.

Finally, since by induction $G_1$ and $G_2$ can be collapsed onto their spines,
and since the composition of homomorphisms is a homomorphism, $G$ can be
collapsed onto its spine.
\end{proof}

Recall that our objective is showing that every series-parallel graph is
constructible.

\begin{lemma}
\label{lem:series-parallel-is-constructible}
Let $(G, u, v)$ be an oriented series-parallel graph. Then, the spine $\cS(G)$
can be extended to $G$ using the doubling and collapsing operations.
Furthermore, the construction preserves the following invariant:
\begin{itemize}
  \item In every doubling step, the doubled vertices on the spine form
    its contiguous (possibly empty) subsegment.
\end{itemize}
\end{lemma}

Since the spine of $G$ can be constructed by repeated application of
Lemma \ref{lem:construct-leaf}, Lemma \ref{lem:series-parallel-is-constructible}
implies what we want.
In the remainder we prove Lemma \ref{lem:series-parallel-is-constructible}
after establishing a couple of technical preliminaries.

\begin{remark}
\label{rem:direct-top-bottom}
In the proof of Lemma \ref{lem:series-parallel-is-constructible} we will
use the fact that whenever $G$ is a composition of
$G_1$ and $G_2$, the edges of $G_1$ and $G_2$ are disjoint.

This is not true if $G$ is a parallel composition 
and there exists a direct edge from top to bottom
in both $G_1$ and $G_2$, but any series-parallel $G$ can be constructed
without using this special case. 
\end{remark}

\begin{claim}
\label{cl:parallel-standard-form}
Let $(G, u, v)$ be an oriented series-parallel graph which is a parallel
composition. Then, there exists a series-parallel construction of 
$(G, u, v)$ such that its final step is a parallel composition of
$G_1$ and $G_2$ with
the following properties:
\begin{itemize}
  \item $L(G_1) \le L(G_2)$.
  \item $G_2$ is a series composition.
\end{itemize}
\end{claim}

\begin{proof}
Firstly, note that whenever $(G, u, v)$ is a generalized composition where
the primary graph is a series or parallel composition, the order of
those two compositions can be reversed without changing the final graph.
Therefore, we can assume w.l.o.g.~that whenever a graph is a generalized
composition, its primary graph has spine of length one.

Let $G$ be a parallel composition of $H'_1$ and $H'_2$. If any of $H'_1$ or
$H'_2$ is a parallel composition, recursively decompose them further until
we are left with a collection of graphs $H_1, \ldots, H_k$ which are all
series or generalized compositions or single edges.

Note that if we compose in parallel $H_1, \ldots, H_k$ in an arbitrary order,
the end result will always be $G$. 

Therefore, we can set $G_2$ 
as $H_i$ with the longest
spine and the parallel composition of the remaining $H_i$ graphs as $G_1$.
Due to Remark \ref{rem:direct-top-bottom}, the spine of $G_2$ must be longer
than one, and therefore $G_2$ must be a series composition.
\end{proof}

Figure \ref{fig:parallel-standard-form} illustrates the content of
Claim \ref{cl:parallel-standard-form}.

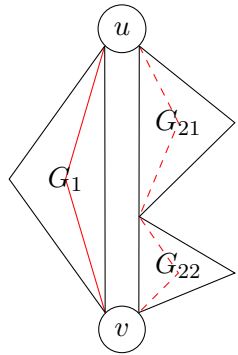
\begin{figure}[ht]\centering
\caption{The continuous red line is the spine of $G$. The dashed red line is
the spine of $G_2$.}
\label{fig:parallel-standard-form}
\begin{tikzpicture}
\node (top) at (0, 0) [draw, circle] {$u$};
\node (bot) at (0, -4) [draw, circle] {$v$};

\draw [red] (top.south west) -- (-0.75, -2) -- (bot.north west);
\draw let \p1 = (top.south east)
in [red, dashed] (\p1) -- (0.75, -1.25) -- (\x1, -2.5) -- (0.75, -3.25)
  -- (bot.north east);

\node at (-0.75, -2) {$G_1$};
\node at (0.75, -1.25) {$G_{21}$};
\node at (0.75, -3.15) {$G_{22}$};

\draw (top.south west) -- (-1.5, -2) -- (bot.north west);
\draw (top.south west) -- (bot.north west);
\draw (top.south east) -- (bot.north east);
\draw let \p1 = (top.south east)
in (\p1) -- (1.5, -1.25) --  (\x1, -2.5) -- (1.5, -3.25) -- (bot.north east);

\end{tikzpicture}
\end{figure}

\begin{proof}[Proof of Lemma \ref{lem:series-parallel-is-constructible}]
Let $(G, u, v)$ be an oriented series-parallel graph. We apply induction
on the number of vertices of $G$ and, secondarily, (in reverse) 
on the length of its spine.

\begin{enumerate}
\item
If $G$ is a single edge, there is nothing to prove (since $\cS(G) = G$).

\item
Assume that $(G, u, v)$ is a series composition of $(G_1, u, w)$ and
$(G_2, w, v)$. Recall that we need to extend $\cS(G)$ to $G$. We do
it in two stages, first extending $\cS(G_1)$ to $G_1$ and then
extending $\cS(G_2)$ to $G_2$.

By induction, we know how to extend $\cS(G_1)$ to $G_1$. Now we will
adapt this sequence of operations to the fact that also $\cS(G_2)$ is present
in the graph. We do it as follows:
\begin{itemize}
\item
  Leave all the collapsing operations as they are (it is always possible
  to collapse onto a bigger graph).
\item
  For doubling operations that keep the vertex $w$ fixed, 
  keep all of $\cS(G_2)$ fixed.
\item
  Finally, let us handle the doubling operations that double the vertex $w$.
  Let $x$ be the neighbour of $w$ on the spine $\cS(G_1)$ and let $y$ be $x$
  in case $x$ is fixed and $x'$ in case $x$ is doubled. Note that the
  edge $(y, w')$ is present in $G_1$ after doubling.

  To emulate this operation in $G$ we double all of $\cS(G_2)$ together with
  $w$ and then collapse the new copy of $\cS(G_2)$ onto the edge $(y, w')$.
\end{itemize}
\begin{figure}\centering
\caption{Handling series decomposition in case $w$ and $x$ are doubled.}
\label{fig:serial}
\begin{tikzpicture}
\tikzset{every node/.style={circle, text=black}}
\node at (0, 0.5) {};

\node (ul) at (0, 0) [draw] {$u$};
\node (xl) at (0, -3) [draw] {$x$};
\node (wl) at (0, -4) [draw] {$w$};
\node (vl) at (0, -6) [draw] {$v$};

\draw (ul) -- (-2, -1.5) -- (wl);
\draw (ul) -- (2, -1.5) -- (wl);

\draw [red] (ul) -- (xl);
\draw [red, thick] (xl) -- (wl);
\draw [red] (wl) -- (vl);

\node at (0, -1.5) {$G_1$};

\node (ur) at (5, 0) [draw] {$u$};
\node (xr) at (5, -3) [draw] {$x$};
\node (wr) at (5, -4) [draw] {$w$};
\node (vr) at (5, -6) [draw] {$v$};
\node (xr') at (5.5, -2) [draw, label=center:$x'$] {\phantom{$x$}};
\node (wr') at (6.5, -3) [draw, label=center:$w'$] {\phantom{$w$}};
\node (vr') at (6.5, -6) [draw, label=center:$v'$] {\phantom{$v$}};

\draw (ur) -- (3, -1.5) -- (wr);
\draw (ur) -- (7, -1.5) -- (wr);

\draw [thick, red] (xr') -- (wr');
\draw [red] (ur) -- (xr);
\draw [red, thick] (xr) -- (wr);
\draw [red] (wr) -- (vr);
\draw [red] (wr') -- (vr');
\draw [thick, green, double] (6.5, -4.5) to [out=0, in=270] (7.75, -3);
\draw [thick, green, double, ->] (7.75, -3) to [out=90, in=45] (6, -2.5);

\node at (5, -1.5) {$G'_1$};
\end{tikzpicture}
\end{figure}
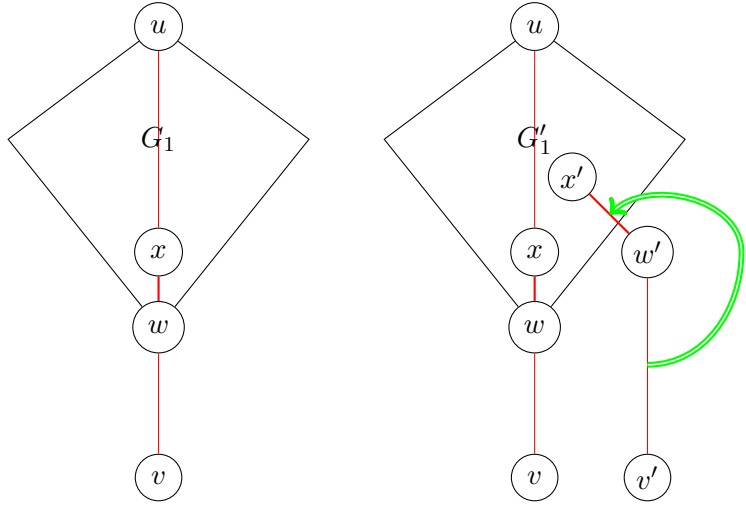
Consult Figure
\ref{fig:serial} for the illustration of one of the cases.

It is easy to see that as a result of this emulation we extend $\cS(G)$
to a series composition of $G_1$ and $\cS(G_2)$.

Now we proceed in the same way to extend $\cS(G_2)$ to $G_2$. The only 
difference is that in case $w$ is doubled we need to double and collapse all 
of $G_1$ instead of just $\cS(G_1)$. This does not pose a problem though,
since $G_1$ can be collapsed onto $\cS(G_1)$ which then can be collapsed
as previously.

Finally, one easily checks that the ``contiguous subsegment'' invariant
of Lemma \ref{lem:series-parallel-is-constructible} is preserved in this
construction.

\item If $(G, u, v)$ is a generalized composition, assume w.l.o.g.~that
it is a composition of the primary graph $(G_1, u, v)$ and $(G_2, v, w)$.
Using Lemma \ref{lem:construct-leaf} we can extend $\cS(G_1)$ to 
$\cS(G_1) \cup \cS(G_2)$ and then proceed as in the series composition case.

\item
Assume that $(G, u, v)$ is a parallel composition of $(G_1, u, v)$ and
$(G_2, u, v)$. By Claim \ref{cl:parallel-standard-form}, we can also assume
that $G_2$ is a series composition of $(G_3, u, w)$ and $(G_4, w, v)$
and that $L(G_1) \le L(G_2)$. In this point we address a subcase where
additionally:
\begin{align}
\label{eq:01a}
L(G_1) + L(G_3) < L(G_4) \; .
\end{align}

$(G, u, v)$ is the parallel composition of $(G_1, u, v)$ and 
the series composition of 
$(G_3, \allowbreak u, \allowbreak w)$ 
and $(G_4, w, v)$. Observe that we can 
also obtain 
$(G, w, v)$ as the parallel composition of $(G_4, w, v)$ and the series \
composition of  $(G_3, w, u)$ and $(G_1, u, w)$. This is illustrated in Figure 
\ref{fig:parallel-rotation}. Furthermore, due to (\ref{eq:01a}) we have
that $L(G, w, v) = L(G_3) + L(G_1) > L(G_1) = L(G, u, v)$.

\begin{figure}\centering
\caption{Rotating $G$ which is a parallel composition.}
\label{fig:parallel-rotation}
\begin{tikzpicture}
\node (0, 0.5) {};

\node (ul) at (0, 0) [draw, circle] {$u$};
\path let \p1=(ul.south east) in node (wl) at (\x1, -1.5) [draw, circle] {$w$};
\node (vl) at (0, -4) [draw, circle] {$v$};

\draw [red] (ul.south west) -- (-0.75, -2) -- (vl.north west);
\draw [red, dashed] (ul.south east) -- (0.75, -0.75) -- (wl) -- 
  (0.75, -2.75) -- (vl.north east);

\node at (-0.75, -2) {$G_1$};
\node at (0.75, -0.75) {$G_{3}$};
\node at (0.75, -2.75) {$G_{4}$};

\draw (ul.south west) -- (-1.5, -2) -- (vl.north west);
\draw (ul.south west) -- (vl.north west);
\draw (ul.south east) -- (wl) -- (vl.north east);
\draw (ul.south east) -- (1.5, -0.75) -- (wl) -- (1.5, -2.75) -- 
(vl.north east);

\node (wr) at (5, 0) [draw, circle] {$w$};
\path let \p1=(wr.south west) in node (ur) at (\x1, -1.5) [draw, circle] {$u$};
\node (vr) at (5, -4) [draw, circle] {$v$};

\draw [red] (wr.south west) -- (4.25, -0.75) -- (ur) -- (4.25, -2.75) --
  (vr.north west);
\draw [red, dashed] (wr.south east) -- (5.75, -2) -- (vr.north east);

\node at (4.25, -0.75) {$G_3$};
\node at (4.25, -2.75) {$G_1$};
\node at (5.75, -2) {$G_4$};

\draw (wr.south west) -- (3.5, -0.75) -- (ur) -- (3.5, -2.75) -- 
  (vr.north west);
\draw (wr.south west) -- (ur) -- (vr.north west);
\draw (wr.south east) -- (vr.north east);
\draw (wr.south east) -- (6.5, -2) -- (vr.north east);
\end{tikzpicture}
\end{figure}
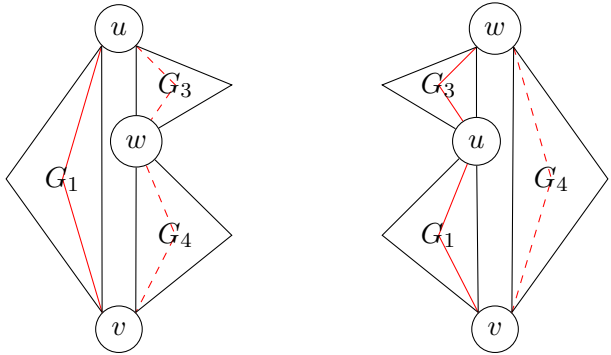

To extend $\cS(G, u, v) = \cS(G_1)$ we proceed as follows:
first, add $\cS(G_3)$ on top of $\cS(G_1)$ using Lemma \ref{lem:construct-leaf}.
Then, extend $\cS(G_3) \cup \cS(G_1) = \cS(G, w, v)$ to $G$ using induction
(which is applicable since the length of the spine increased).

Again, one easily checks that the contiguous subsegment invariant is preserved
in this construction.

\item
If $G$ is a parallel composition and $L(G_1)+L(G_4) < L(G_3)$, we proceed
symmetrically as in case 4.

\item
Finally, let $(G, u, v)$ be a parallel composition and:
\begin{align}
L(G_1) + L(G_3) \ge L(G_4) \label{eq:02a} \; , \\
L(G_1) + L(G_4) \ge L(G_3) \label{eq:03a} \; .
\end{align}
Again, we proceed in stages successively building $(G_1, u, v)$, 
$(G_3, u, w)$ and 
$(G_4, w, v)$ using induction.

We start with $\cS(G) = \cS(G_1)$, which we need to extend to $G$.
First, by induction we extend $\cS(G_1)$ to $G_1$. Next, we add 
$\cS(G_2) = \cS(G_3) \cup \cS(G_4)$ as follows: let $a := L(G_1)$
and $b := L(G_2)$. Recall that $b \ge a$ and that $a$ and $b$ have the same 
parity.

Using Lemma \ref{lem:construct-leaf}, add a path of length $(b-a)/2$ starting
from $u$ and let $x$ be the endpoint of this path. Fix $v$ and $x$
and double all the other vertices. Finally, collapse the resulting 
copy of $G_1$ onto the path from $v$ to $u'$.

In the next stage, we work with the sequence that extends $\cS(G_3)$ to $G_3$. 
We need to adapt it to additional edges we have in the graph. This is done as
follows:
\begin{itemize}
\item All collapsing operations stay the same.
\item Doubling operations that keep both $u$ and $w$ fixed fix all the
vertices of $G_1$ and $\cS(G_4)$.
\item In case at least one of $u$ and $w$ is doubled the arguments are 
very similar to each other. Therefore we present only the one where
$u$ is doubled and $w$ is fixed. See Figure \ref{fig:parallel-doubling}
for a graphical illustration.

\begin{figure}\centering
\caption{Handling parallel decomposition when segment from $u$ to $y$ is
doubled. For clarity, $G_1$ and $G_4$ are drawn as spines only. 
The blue path is collapsed onto the green path.}
\label{fig:parallel-doubling}
\begin{tikzpicture}
\node at (0, 0.5) {};

\node [draw, circle] (ur) at (0, 0) {$u$};
\node [draw, circle] (vr) at (0, -4) {$v$};
\node [draw, circle] (wr) at (1, -3.5) {$w$};
\node [draw, circle] (xr) at (1.25, -2) {$x$};
\node [draw, circle] (yr) at (1, -1) {$y$};
\node [draw, circle, label=center:$y'$] (yr') at (1.75, -1.25) {\phantom{$y$}};
\node [draw, circle, label=center:$u'$] (ur') at (1.75, -0.35) {\phantom{$u$}};
\node [draw, circle, label=center:$v'$] (vr') at (2, -4) {\phantom{$v$}};

\draw [very thick] (ur) -- (wr) -- (3.5, 0) -- (ur);
\draw [red] (ur) -- (vr) -- (wr);
\draw [green] (wr)  -- (xr);
\draw [red, thick] (xr) -- (yr);
\draw [red] (yr) -- (ur);
\draw [green, thick] (xr) to (yr');
\draw [green] (yr') -- (ur');
\draw [blue] (ur') .. controls (5, 1) .. (vr'); 
\draw [blue] (vr') -- (wr);
\draw [thick, double, green, ->] (3.1, -2) -- (1.6, -1.75);

\node at (2.6, -0.5) {$G'_3$};

\node [draw, circle] (ul) at (-5, 0) {$u$};
\node [draw, circle] (vl) at (-5, -4) {$v$};
\node [draw, circle] (wl) at (-4, -3.5) {$w$};
\node [draw, circle] (xl) at (-3.75, -2) {$x$};
\node [draw, circle] (yl) at (-4, -1) {$y$};

\draw [very thick] (ul) -- (wl) -- (-1.5, 0) -- (ul);
\draw [red] (ul) -- (vl) -- (wl);
\draw [red] (wl)  -- (xl);
\draw [red, thick] (xl) -- (yl);
\draw [red] (yl) -- (ul);

\node at (-2.4, -0.5) {$G_3$};
\end{tikzpicture}
\end{figure}
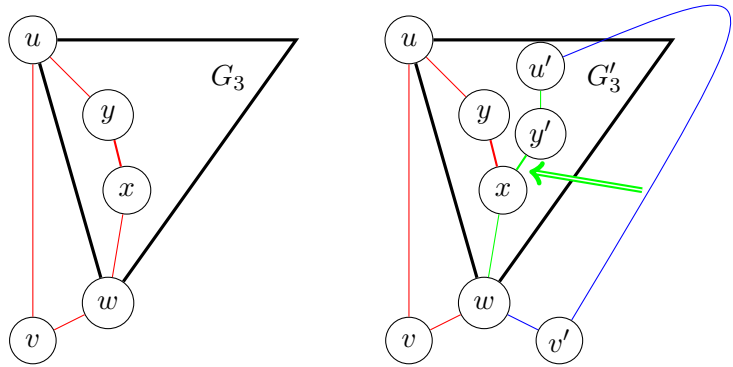

By inductive assumption, we know that a contiguous subpath of the spine 
$\cS(G_3)$ is doubled. Assume that its doubled vertices go from $u$ to $y$
and the fixed ones from $w$ to $x$ (i.e., $x$ and $y$ are neighbours on the
spine).

To emulate this case in $G$, double all vertices of $G_1$ and $\cS(G_4)$
except of $w$. Next, collapse the new copy of $G_1$ onto its spine.
Finally, collapse the resulting path $\cP_1 := u'-v'-w$ onto the copy
of $\cS(G_3)$, i.e., $\cP_2 := u'-y'-x-w$. This is possible due
to~\eqref{eq:03a}: since the path $\cP_1$ is at least as long
as $\cP_2$, $\cP_1$ can be collapsed onto $\cP_2$ as in the proof
of Lemma \ref{lem:spine-collapse}.
\end{itemize}

Finally, we construct $G_4$ from $\cS(G_4)$ in a very similar way. The only
differences are that when emulating doubling we need to perform an additional 
collapse of $G_3$ onto $\cS(G_3)$ and that we rely on the 
inequality~\eqref{eq:02a} for the final collapse.

Again, one checks that the contiguous subsegment invariant is preserved
throughout the whole process.
\end{enumerate}
\end{proof}

\section{Some Graphs Are Not Constructible}
\label{sec:non-constructible}

It is an open question if all two-prover question sets admit exponential
parallel repetition. One way to prove that they do would be to show
that all graphs are constructible by conditioning.
However, in this section we show that that is not
the case, hence another way must be found to resolve this open question:

\begin{definition}
Let $n \in \bbN$ be even and greater or equal to $8$. 
We define the \emph{cycle with shortcuts}
$\mathfrak{C}_n$ as the following simple graph:
$V(\mathfrak{C}_n) := \{0, \ldots, n-1\}$ and $\{u, v\} \in E(\mathfrak{C}_n)$
if and only if $|u-v| \in \{1,3,n-3,n-1\}$.
\end{definition}

See Figure \ref{fig:cycle_12} for a drawing of $\mathfrak{C}_{12}$. 
Observe that $\mathfrak{C}_n$ is bipartite. We show:

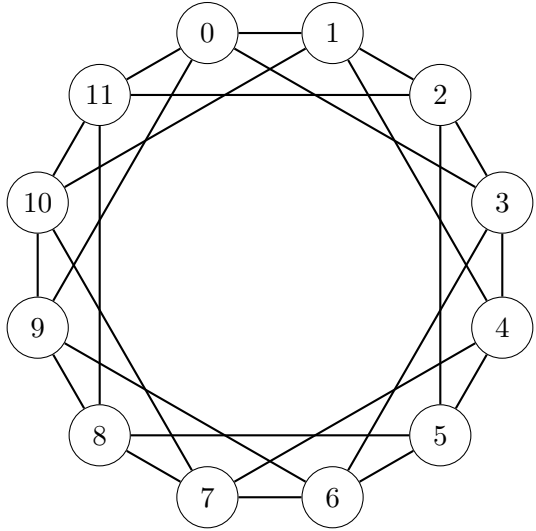
\begin{figure}\centering
\caption{A drawing of $\mathfrak{C}_{12}$.}
\label{fig:cycle_12}
\begin{tikzpicture}[scale=0.8]
\node at (0, 4.36) {};

\node [draw, circle, label=center:$0$] (u0) at (-1.04, 3.86) {\phantom{$00$}};
\node [draw, circle, label=center:$1$] (u1) at (1.04, 3.86) {\phantom{$00$}};
\node [draw, circle, label=center:$2$] (u2) at (2.83, 2.83) {\phantom{$00$}};
\node [draw, circle, label=center:$3$] (u3) at (3.86, 1.04) {\phantom{$00$}};
\node [draw, circle, label=center:$4$] (u4) at (3.86, -1.04) {\phantom{$00$}};
\node [draw, circle, label=center:$5$] (u5) at (2.83, -2.83) {\phantom{$00$}};
\node [draw, circle, label=center:$6$] (u6) at (1.04, -3.86) {\phantom{$00$}};
\node [draw, circle, label=center:$7$] (u7) at (-1.04, -3.86) {\phantom{$00$}};
\node [draw, circle, label=center:$8$] (u8) at (-2.83, -2.83) {\phantom{$00$}};
\node [draw, circle, label=center:$9$] (u9) at (-3.86, -1.04) {\phantom{$00$}};
\node [draw, circle] (u10) at (-3.86, 1.04) {$10$};
\node [draw, circle] (u11) at (-2.83, 2.83) {$11$};

\draw [thick] (u0) -- (u1) -- (u2) -- (u3) -- (u4) -- (u5) -- (u6) -- (u7) --
(u8) -- (u9) -- (u10) -- (u11) -- (u0);
\draw [thick] (u0) -- (u3) -- (u6) -- (u9) -- (u0);
\draw [thick] (u1) -- (u4) -- (u7) -- (u10) -- (u1);
\draw [thick] (u2) -- (u5) -- (u8) -- (u11) -- (u2);
\end{tikzpicture}
\end{figure}

\begin{theorem}[cf.~Theorem~\ref{thm:non-constructible-simple}]
\label{thm:cycle-non-constructible}
The cycle with shortcuts $\mathfrak{C}_{12}$ is not constructible by conditioning.
\end{theorem}

Since any bipartite graph $G$ joined with $\mathfrak{C}_{12}$ by a single vertex can be 
collapsed onto $\mathfrak{C}_{12}$, Theorem~\ref{thm:cycle-non-constructible}
implies the existence of an infinite family of graphs that
are not constructible.

Our proof of Theorem~\ref{thm:cycle-non-constructible} turns out to be
somewhat involved and computer-assisted. 
Before we proceed with it, we explain why another natural proof idea fails.

\subsection{``Warm-up'': constructing all induced subgraphs}
\label{sec:all-induced}

A natural idea to prove Theorem~\ref{thm:cycle-non-constructible}
would be to show for a certain graph $G$ that if it is not already present
as an induced subgraph in another graph $H$, then no doubling of $H$
can produce an induced instance of $G$. It turns out
that this approach must fail, since for every bipartite graph $G$ we
can construct a graph $H$ such that $G$ is an induced subgraph of $H$.

\begin{definition}
Let $k \ge 1$. We define the \emph{set graph} $\fS_k := (X, Y, E)$ as
follows:
\begin{itemize}
  \item $X := [k]$.
  \item $Y := \left\{ S \subset [k]: S \ne \emptyset \right\}$.
  \item $E := \left\{ (x, S): x \in S  \right\}$.
\end{itemize}
\end{definition}

\begin{theorem}
The set graph $\fS_k$ is constructible by conditioning  with $2(k-1)$ doublings.
\end{theorem}

\begin{proof}
The proof is by induction on $k$. The graph $\fS_1$ is just a single edge.
To construct $\fS_{k+1}$, start with constructing
$\fS_k$ with $2(k-1)$ doublings.

We make a preliminary point to avoid confusion.
Note that the right hand-side vertices of $\fS_k$ are labeled with
subsets of $[k]$ such that for a vertex labeled with $S$ we have that
its neighborhood is equal to its label: $N(S) = S$.
We will now perform some doublings and label the new vertices with subsets
that contain $k+1$. However, for a new vertex with a label $S$ it is not
evident that $N(S) = S$: this is what we have to prove.

After constructing $\fS_k$, 
perform a doubling as follows: double all vertices labeled with $S$ such that
$k \in S$ and label each new vertex as $S \cup \{k+1\}$.

Then, perform a second doubling: double $k$ and, again, all vertices
labeled with $S$
such that $k \in S$ and $k+1 \notin S$. This time label the copy of $k$ as $k+1$
and a copy of $S$ as $S \setminus \{k\} \cup \{k+1\}$.

Note that after the doublings $Y = \{ S \subseteq [k+1]: S \ne \emptyset\}$.
For $S \in Y$ let $N(S) := \{x \in X: (x,S) \in E\}$ be the neighborhood
of $S$. We need to check that $N(S) = S$ for every label $S$. This holds by
the following case analysis:
\begin{itemize}
  \item Each vertex labeled with $S$ such that $k+1 \notin S$ existed before
    the first doubling and its neighborhood did not change
    (since it was doubled in the second doubling in case $k \in S$). 

  \item
    Each vertex labeled with $S$ such that $\{k, k+1\} \subseteq S$ was
    created in the first doubling, at which point we had
    $N(S) = S \setminus \{k+1\}$. Then, it was fixed in the second
    doubling and $k+1$ was added to its neighborhood.

  \item Each vertex labeled with $S$ such that $k \notin S$ and $k+1 \in S$
    was created in the second doubling with $N(S) = S$.
\end{itemize}
Therefore, we can construct $\fS_{k+1}$ from $\fS_{k}$ in $2$ doublings
and $\fS_{k+1}$ from $\fS_1$ in $2k$ doublings.
\end{proof}

\begin{remark}
  A modification of this construction can be used to construct $\fS_{k,r}$
  with $X := [k]$, $Y := \{S \subseteq [k]: |S| = r\}$ and 
  $E := \{(x,S): x \in S\}$.
\end{remark}

Now we turn to the proof of Theorem~\ref{thm:cycle-non-constructible}.

\subsection{Decomposing last two steps}
\label{sec:trivial-doublings}

\begin{definition}
Let $u, v$ be two vertices arising during a construction of a bipartite
graph $G$. We write $u \sim v$ if $u$ and $v$ are adjacent. 
For two sets of vertices $A, B$, we
write $E(A, B)$ for the set of edges between $A$ and $B$.
We also write $G(A)$ for the graph induced by vertices in $A$.
\end{definition}

Note that the operators $\sim$, $E(\cdot, \cdot)$ and $G(\cdot)$ 
do not depend on the stage of the construction:
doubling and collapsing only add and remove vertices, without
changing existing adjacencies. 

\begin{lemma}
\label{lem:one-collapse}
Let $G$ be bipartite graph. If $G$ is constructible, then it is constructible
such that all the operations except for the last one are doublings.
\end{lemma}
\begin{proof}
First, assume that in a construction of $G$ there is a collapse operation
immediately followed by a doubling operation. Assume that $A$ is the set
of the vertices collapsed in the first operation, 
$B$ is the set of vertices that are fixed in the first
operation and doubled onto $B$' in the second operation and $C$ the set
of vertices that are fixed throughout both operations 
(see Figure~\ref{fig:exchange-collapse-doubling}).

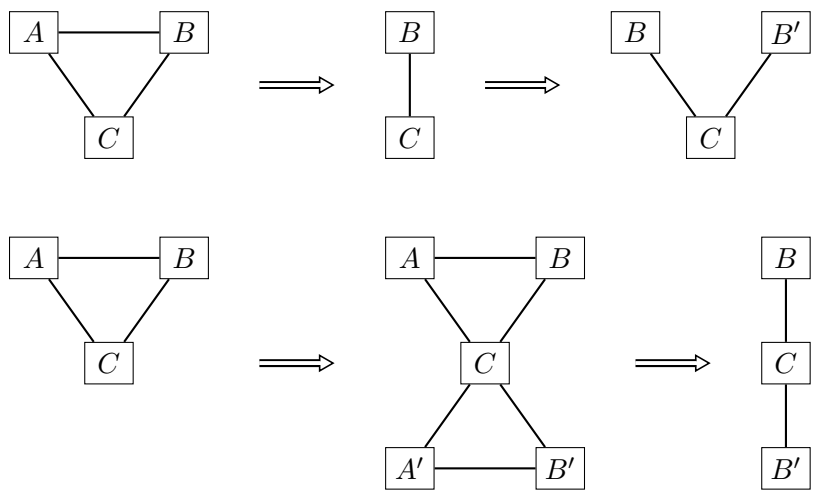
\begin{figure}[ht]\centering
\caption{Transposing a collapse and a doubling.}
\label{fig:exchange-collapse-doubling}
\begin{tikzpicture}
\node at (0, 0.5) {};

\node [draw, rectangle, label=center:$A$] (A1) at (0, 0) {\phantom{$A'$}};
\node [draw, rectangle, label=center:$B$] (B1) at (2, 0) {\phantom{$A'$}};
\node [draw, rectangle, label=center:$C$] (C1) at (1, -1.4) {\phantom{$A'$}};
\draw [thick] (A1) -- (B1) -- (C1) -- (A1);

\draw [vecArrow] (3, -0.7) -- (4, -0.7);

\node [draw, rectangle, label=center:$B$] (B2) at (5, 0) {\phantom{$A'$}};
\node [draw, rectangle, label=center:$C$] (C2) at (5, -1.4) {\phantom{$A'$}};
\draw [thick] (B2) -- (C2);

\draw [vecArrow] (6, -0.7) -- (7, -0.7);

\node [draw, rectangle, label=center:$B$] (B3) at (8, 0) {\phantom{$A'$}};
\node [draw, rectangle, label=center:$B'$] (B'3) at (10, 0) {\phantom{$A'$}};
\node [draw, rectangle, label=center:$C$] (C3) at (9, -1.4) {\phantom{$A'$}};
\draw [thick] (B3) -- (C3) -- (B'3);

\node [draw, rectangle, label=center:$A$] (A4) at (0, -3) {\phantom{$A'$}};
\node [draw, rectangle, label=center:$B$] (B4) at (2, -3) {\phantom{$A'$}};
\node [draw, rectangle, label=center:$C$] (C4) at (1, -4.4) {\phantom{$A'$}};
\draw [thick] (A4) -- (B4) -- (C4) -- (A4);

\draw [vecArrow] (3, -4.4) -- (4, -4.4);

\node [draw, rectangle, label=center:$A$] (A5) at (5, -3) {\phantom{$A'$}};
\node [draw, rectangle, label=center:$B$] (B5) at (7, -3) {\phantom{$A'$}};
\node [draw, rectangle, label=center:$C$] (C5) at (6, -4.4) {\phantom{$A'$}};
\node [draw, rectangle, label=center:$A'$] (A'5) at (5, -5.8) {\phantom{$A'$}};
\node [draw, rectangle, label=center:$B'$] (B'5) at (7, -5.8) {\phantom{$A'$}};
\draw [thick] (A5) -- (B5) -- (C5) -- (A5);
\draw [thick] (A'5) -- (B'5) -- (C5) -- (A'5);

\draw [vecArrow] (8, -4.4) -- (9, -4.4);

\node [draw, rectangle, label=center:$B$] (B6) at (10, -3) {\phantom{$A'$}};
\node [draw, rectangle, label=center:$C$] (C6) at (10, -4.4) {\phantom{$A'$}};
\node [draw, rectangle, label=center:$B'$] (B'6) at (10, -5.8) {\phantom{$A'$}};
\draw [thick] (B6) -- (C6) -- (B'6);
\end{tikzpicture}
\end{figure}

Then, those two operations can be exchanged as follows.
First, double $A$ onto $A'$ and $B$ onto $B'$. Then, collapse
$A$ onto $B \cup C$ and $A'$ onto $B' \cup C$ 
(again see Figure~\ref{fig:exchange-collapse-doubling}).
In both cases we end up with the
same graph on vertices $B \cup B' \cup C$.

Finally, note that once all collapses are at the end of the
sequence of the operations, they can be merged into a single collapse.
\end{proof}

\begin{definition}
We say that a graph $G$ is \emph{collapsible} onto a graph $H$,
if $H$ can be constructed from $G$ by a single collapse operation.
\end{definition}

\begin{lemma}
\label{lem:construction-collapsible}
Let $H$ be a constructible graph with at least two edges. 
There exists a construction of $H$ such that:
\begin{enumerate}
  \item The last operation is a collapse.
  \item All other operations are doublings.
  \item Leting $H_0$ be the graph before the last doubling, $H_0$ is not 
    collapsible onto $H$.
\end{enumerate}
\end{lemma}

\begin{proof}
By Lemma \ref{lem:one-collapse}, there exists a construction of $H$
satisfying the first two conditions. 
Let us take such a construction with the smallest
possible number of doublings. Since $H$ is not a single
edge, the number of doublings must be at least one.

If in this construction
$H_0$ is collapsible onto $H$, the last doubling and the collapse can
be replaced with a single collapse, which is a contradiction.
\end{proof}

Due to Lemma \ref{lem:one-collapse}, 
we can assume that if the graph $\mathfrak{C}_{12}$ is
constructible, the last two steps of its construction are, respectively, 
doubling and collapsing.
Let us now divide the vertices of the construction
depending on what happens to them in those last two steps 
(see Figure \ref{fig:last-two-steps}).

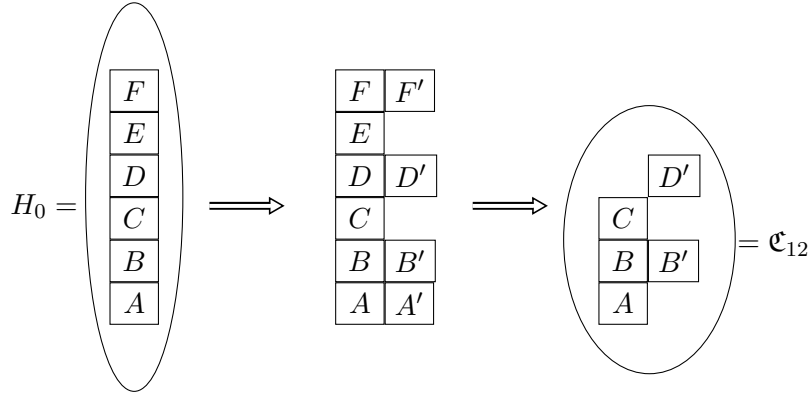
\begin{figure}[ht]\centering
\caption{The last two steps in a construction of $\mathfrak{C}_{12}$.}
\label{fig:last-two-steps}
\begin{tikzpicture}
\node at (0, 3.7) {};

\node [draw, rectangle, label=center:$A$] (A1) at (0, 0) {\phantom{$A'$}};
\node [draw, rectangle, label=center:$B$, above=0cm of A1] (B1) {\phantom{$A'$}};
\node [draw, rectangle, label=center:$C$, above=0cm of B1] (C1) {\phantom{$A'$}};
\node [draw, rectangle, label=center:$D$, above=0cm of C1] (D1) {\phantom{$A'$}};
\node [draw, rectangle, label=center:$E$, above=0cm of D1] (E1) {\phantom{$A'$}};
\node [draw, rectangle, label=center:$F$, above=0cm of E1] (F1) {\phantom{$A'$}};

\draw [vecArrow] (1, 1.3) -- (2, 1.3);

\node [draw, rectangle, label=center:$A$] (A2) at (3, 0) {\phantom{$A'$}};
\node [draw, rectangle, label=center:$B$, above=0cm of A2] (B2) {\phantom{$A'$}};
\node [draw, rectangle, label=center:$C$, above=0cm of B2] (C2) {\phantom{$A'$}};
\node [draw, rectangle, label=center:$D$, above=0cm of C2] (D2) {\phantom{$A'$}};
\node [draw, rectangle, label=center:$E$, above=0cm of D2] (E2) {\phantom{$A'$}};
\node [draw, rectangle, label=center:$F$, above=0cm of E2] (F2) {\phantom{$A'$}};
\node [draw, rectangle, right=0cm of A2] (A'2) {$A'$};
\node [draw, rectangle, right=0cm of B2] (B'2) {$B'$};
\node [draw, rectangle, right=0cm of D2] (D'2) {$D'$};
\node [draw, rectangle, right=0cm of F2] (F'2) {$F'$};

\draw [vecArrow] (4.5, 1.3) -- (5.5, 1.3);

\node [draw, rectangle, label=center:$A$] (A3) at (6.5, 0) {\phantom{$A'$}};
\node [draw, rectangle, label=center:$B$, above=0cm of A3] (B3) {\phantom{$A'$}};
\node [draw, rectangle, label=center:$C$, above=0cm of B3] (C3) {\phantom{$A'$}};
\node [draw, rectangle, right=0cm of B3] (B'3) {$B'$};
\node [draw, rectangle, above right=0cm of C3] (D'3) {$D'$};

\node [draw, shape=ellipse, fit=(A1)(B1)(C1)(D1)(E1)(F1)] {};
\node [draw, shape=ellipse, fit=(A3)(B3)(C3)(B'3)(D'3)] {};

\node at (-1.2, 1.3) {$H_0 =$};
\node at (8.5, 0.8) {$= \mathfrak{C}_{12}$};
\end{tikzpicture}
\end{figure}

The division is as follows: $A$ are vertices that are doubled onto $A'$
in the first step, with $A$ fixed and $A'$ collapsed in the second step. 
$B$ are vertices doubled onto $B'$ in the first step with
both $B$ and $B'$ fixed in the second step. $C$ are vertices that
are fixed throughout both steps. $D$ are vertices doubled onto $D'$ in the
first step with $D$ collapsed and $D'$ fixed in the second step.
$E$ are vertices fixed in the first step and collapsed in the second step.
Finally, $F$ are vertices that are doubled onto $F'$ in the first step
with both $F$ and $F'$ collapsed in the second step.

One checks that this division covers all possible events in the last two
steps.
The final graph $\mathfrak{C}_{12}$ consists of vertices 
$A \cup B \cup B' \cup C \cup D'$.

Our proof of Theorem \ref{thm:cycle-non-constructible} goes as follows:
First, we show that if the last two steps of a construction of 
$\mathfrak{C}_{12}$ are as above, it must be $B = \emptyset$
and $E(A, D) = \emptyset$. Then, we prove that if $B = \emptyset$
and $E(A, D) = \emptyset$, then the initial graph $H_0$ must have been
collapsible onto $\mathfrak{C}_{12}$ in the first place.
Parts of the proof are computer-assisted, with the codes of C++
programs provided in Appendix \ref{ch:appendix}.

\subsection{Non-collapsible graphs never produce 
\texorpdfstring{$\mathfrak{C}_{12}$}{C\_12}
}

\begin{lemma}
\label{lem:assisted-non-empty-b}
Let $\mathfrak{C}_{12}$ be constructed in two steps
from some bipartite $H_0$, as above.
It cannot be that $E = F = \emptyset$, $E(A, D) = \emptyset$ and 
$B \ne \emptyset$.
\end{lemma}

\begin{proof}
Computer-assisted (enumerate all partitions of $\mathfrak{C}_{12}$
into $A \cup B \cup B' \cup C \cup D'$ together with a bijection
between $B$ and $B'$, since $E(A, D) = \emptyset$ such a partition
implies a unique $H_0 = G(A \cup B \cup C \cup D)$), 
see the program \verb+non_empty_b.cpp+
in Listing \ref{lst:non-empty-b}.
\end{proof}

\begin{lemma}
\label{lem:assisted-non-empty-ad}
Let $\mathfrak{C}_{12}$ be constructed in two steps
from some bipartite $H_0$, as above.
It cannot be that $B = E = F = \emptyset$ and $|E(A, D)| = 1$.
\end{lemma}

\begin{proof}
Computer-assisted (enumerate all partitions of $\mathfrak{C}_{12}$ into
$A \cup C \cup D'$ and all edges between $A$ and $D$, again
this implies a unique $H_0 = G(A \cup C \cup D)$),
see the program \verb+non_empty_ad.cpp+ in Listing \ref{lst:non-empty-ad}.
\end{proof}

\begin{lemma}
\label{lem:only-trivial-doublings}
Let $\mathfrak{C}_{12}$ be constructed in two steps from some bipartite $H_0$,
as above. Then, it must be that $B = \emptyset$ and $E(A, D) = \emptyset$.
\end{lemma}

\begin{proof}
Assume by contradiction that there exists a construction of $\mathfrak{C}_{12}$
with $B \ne \emptyset$ or $E(A, D) \ne \emptyset$.

Firstly, note that the same construction but with the vertices from
$E \cup F$ deleted from the initial graph $H_0$
is valid and also results in $\mathfrak{C}_{12}$. Therefore,
we can assume w.l.o.g.~that $E = F = \emptyset$.

We now proceed in two cases. If $B \ne \emptyset$, we can additionally
assume that $E(A, D) = \emptyset$. This is again due to the fact that
if we deleted $E(A, D)$ edges from $H_0$, we would still obtain
a valid construction that results in $\mathfrak{C}_{12}$ 
(cf.~Figure \ref{fig:only-trivial-b-non-empty}). But $B \ne \emptyset$
and $E(A, D) = \emptyset$ is impossible due to Lemma
\ref{lem:assisted-non-empty-b}.

On the other hand, assume that $B = \emptyset$ and $E(A, D) \ne \emptyset$.
Then, by the same argument as before, we can also assume that the size of 
$E(A, D)$ is as small as possible, namely $|E(A, D)| = 1$
(cf.~Figure \ref{fig:only-trivial-b-empty}). But this also yields
a contradiction by Lemma \ref{lem:assisted-non-empty-ad}.
\end{proof}

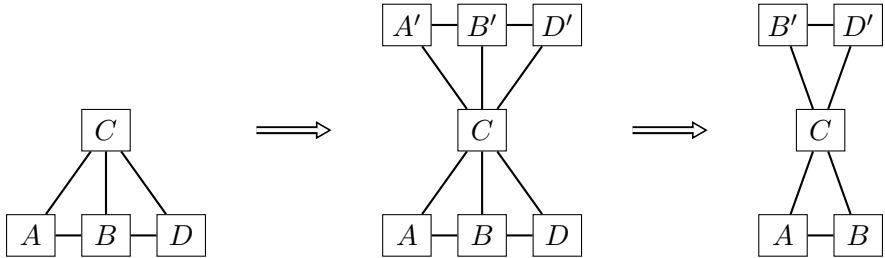
\begin{figure}[ht]\centering
\caption{An illustration of Lemma \ref{lem:only-trivial-doublings},
case $B \ne \emptyset$.}
\label{fig:only-trivial-b-non-empty}
\begin{tikzpicture}
\node at (0, 1.9) {};

\node [draw, rectangle, label=center:$C$] (C1) at (0, 0) {\phantom{$A'$}};
\node [draw, rectangle, label=center:$A$] (A1) at (-1, -1.4) {\phantom{$A'$}};
\node [draw, rectangle, label=center:$B$] (B1) at (0, -1.4) {\phantom{$A'$}};
\node [draw, rectangle, label=center:$D$] (D1) at (1, -1.4) {\phantom{$A'$}};
\draw [thick] (C1) -- (A1) -- (B1) -- (D1) -- (C1) -- (B1);

\draw [vecArrow] (2, 0) -- (3, 0);

\node [draw, rectangle, label=center:$A'$] (A'2) at (4, 1.4) {\phantom{$A'$}};
\node [draw, rectangle, label=center:$B'$] (B'2) at (5, 1.4) {\phantom{$A'$}};
\node [draw, rectangle, label=center:$D'$] (D'2) at (6, 1.4) {\phantom{$A'$}};
\node [draw, rectangle, label=center:$C$] (C2) at (5, 0) {\phantom{$A'$}};
\node [draw, rectangle, label=center:$A$] (A2) at (4, -1.4) {\phantom{$A'$}};
\node [draw, rectangle, label=center:$B$] (B2) at (5, -1.4) {\phantom{$A'$}};
\node [draw, rectangle, label=center:$D$] (D2) at (6, -1.4) {\phantom{$A'$}};
\draw [thick] (B'2) -- (C2) -- (A'2) -- (B'2) -- (D'2) -- (C2) -- (A2) -- (B2)
-- (D2) -- (C2) -- (B2);

\draw [vecArrow] (7, 0) -- (8, 0);

\node [draw, rectangle, label=center:$B'$] (B'3) at (9, 1.4) {\phantom{$A'$}};
\node [draw, rectangle, label=center:$D'$] (D'3) at (10, 1.4) {\phantom{$A'$}};
\node [draw, rectangle, label=center:$C$] (C3) at (9.5, 0) {\phantom{$A'$}};
\node [draw, rectangle, label=center:$A$] (A3) at (9, -1.4) {\phantom{$A'$}};
\node [draw, rectangle, label=center:$B$] (B3) at (10, -1.4) {\phantom{$A'$}};
\draw [thick] (C3) -- (B'3) -- (D'3) -- (C3) -- (B3) -- (A3) -- (C3);

\end{tikzpicture}
\end{figure}

\begin{figure}[ht]\centering
\caption{An illustration of Lemma \ref{lem:only-trivial-doublings},
case $B = \emptyset$, $|E(A, D)| = 1$. The edge between $A$ and $D$
is marked red.}
\label{fig:only-trivial-b-empty}
\begin{tikzpicture}
\node at (0, 1.9) {};

\node [draw, rectangle, label=center:$A$] (A1) at (0, 1.4) {\phantom{$A'$}};
\node [draw, rectangle, label=center:$C$] (C1) at (0, 0) {\phantom{$A'$}};
\node [draw, rectangle, label=center:$D$] (D1) at (0, -1.4) {\phantom{$A'$}};
\draw [thick] (A1) -- (C1) -- (D1);
\draw [thick, red] (A1.west) to [bend right] (D1.west);

\draw [vecArrow] (1, 0) -- (2, 0);

\node [draw, rectangle, label=center:$A$] (A2) at (3, 1.4) {\phantom{$A'$}};
\node [draw, rectangle, label=center:$A'$] (A'2) at (4, 1.4) {\phantom{$A'$}};
\node [draw, rectangle, label=center:$C$] (C2) at (3.5, 0) {\phantom{$A'$}};
\node [draw, rectangle, label=center:$D$] (D2) at (3, -1.4) {\phantom{$A'$}};
\node [draw, rectangle, label=center:$D'$] (D'2) at (4, -1.4) {\phantom{$A'$}};
\draw [thick] (A2) -- (C2) -- (D2);
\draw [thick] (A'2) -- (C2) -- (D'2);
\draw [thick, red] (A2.west) to [bend right] (D2.west);
\draw [thick, red] (A'2.east) to [bend left] (D'2.east);

\draw [vecArrow] (5, 0) -- (6, 0);

\node [draw, rectangle, label=center:$A$] (A3) at (7, 1.4) {\phantom{$A'$}};
\node [draw, rectangle, label=center:$C$] (C3) at (7, 0) {\phantom{$A'$}};
\node [draw, rectangle, label=center:$D'$] (D'3) at (7, -1.4) {\phantom{$A'$}};
\draw [thick] (A3) -- (C3) -- (D'3);
\end{tikzpicture}
\end{figure}
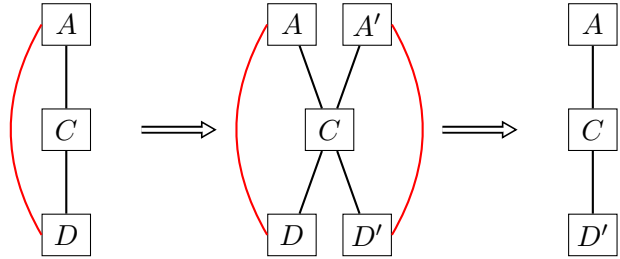

We need some additional concepts
to deal with the remaining case $B = \emptyset$, $E(A, D) = \emptyset$. 

\begin{definition}
Let $A$ and $B$ be disjoint sets of vertices that exist at some point
during a construction of a bipartite graph $H$. Assume that a doubling
operation is performed and that all vertices from $B$ (possibly together
with some vertices from $A$ and outside of $A \cup B$) are doubled.

Let $B'$ be the set of copies of vertices from $B$.
There is an obvious bijection between $B$ and $B'$ which
we call the \emph{natural bijection}. 
Similarly, we say
that there is natural bijection between $A \cup B$
and $A \cup B'$.
If this bijection is also an isomorphism between
$G(A \cup B)$ and $G(A \cup B')$, we say
that $G(A \cup B)$ and $G(A \cup B')$ are
\emph{naturally isomorphic}.
\end{definition}

\begin{definition}
Let $\mathfrak{C}_{12}$ be constructed from some $H_0$ in two steps, as above.
We say that $A'$ was \emph{naturally collapsed}  onto $A$, if in the collapse
step each vertex of $A'$ was collapsed onto $A$ via the natural bijection.
Analogously, we say that $D$ was naturally collapsed onto $D'$.
\end{definition}

\begin{lemma}
\label{lem:assisted-a-d-collapse}
Let $\mathfrak{C}_{12}$ be constructed from some $H_0$
with one doubling and one collapse, as above.
If $B = E = F = \emptyset$ and $E(A, D) = \emptyset$, then
in the subsequent collapse either $A'$ is naturally collapsed onto $A$
or $D$ is naturally collapsed onto $D'$.
\end{lemma}

\begin{proof}
Computer-assisted (enumerate all partitions of $\mathfrak{C}_{12}$ into
$A \cup C \cup D'$ and all possible collapses),
see the program \verb+natural_collapse.cpp+
in Listing \ref{lst:natural}.
\end{proof}

\begin{lemma}
\label{lem:auto-a-d-collapse}
Let $\mathfrak{C}_{12}$ be constructed from some $H_0$
with one doubling and one collapse, as above.
If $B = \emptyset$ and $E(A, D) = \emptyset$, then
in the subsequent collapse either $A'$ is naturally collapsed onto $A$
or $D$ is naturally collapsed onto $D'$.
\end{lemma}

\begin{proof}
Assume there exists a construction of $\mathfrak{C}_{12}$ from 
some $H = G(A \cup C \cup D \cup E \cup F$) such that:
\begin{enumerate}
\item $B = \emptyset$ and $E(A, D) = \emptyset$.
\item $A'$ does not naturally collapse onto $A$ and $D$ does not 
naturally collapse onto $D'$.  
\end{enumerate}
Then, the same construction with vertices $E \cup F$ omitted
from $H_0$
is also valid and satisfies both conditions. 
Therefore, we can assume w.l.o.g.~that
$E = F = \emptyset$.
Then, the result follows from Lemma \ref{lem:assisted-a-d-collapse}.
\end{proof}

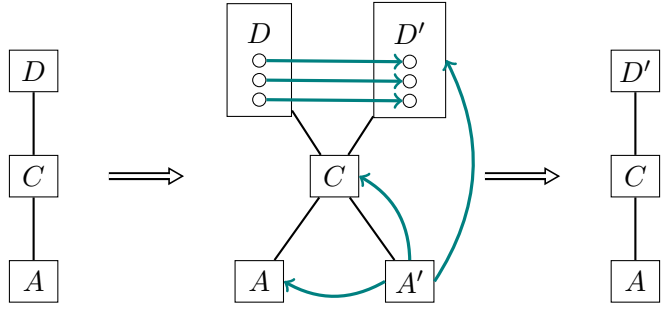
\begin{figure}[ht]\centering
\caption{An illustration of Lemma \ref{lem:auto-a-d-collapse}, case when
$E = F = \emptyset$ and 
$D$ collapses naturally onto $D'$. Blue arrows denote a collapse.}
\label{fig:auto-a-d-collapse}
\begin{tikzpicture}
\node at (0, 0.9) {};

\node [draw, rectangle, label=center:$D$] (D1) at (0, 0) {\phantom{$A'$}};
\node [draw, rectangle, label=center:$C$] (C1) at (0, -1.4) {\phantom{$A'$}};
\node [draw, rectangle, label=center:$A$] (A1) at (0, -2.8) {\phantom{$A'$}};
\draw [thick] (D1) -- (C1) -- (A1);

\draw [vecArrow] (1, -1.4) -- (2, -1.4);

\node (D2a) at (3, 0.5) {$D$};
\node [draw, circle, inner sep=0pt, minimum size=5pt, below=0cm of D2a] (D2b) {};
\node [draw, circle, inner sep=0pt, minimum size=5pt, below=2pt of D2b] (D2c) {};
\node [draw, circle, inner sep=0pt, minimum size=5pt, below=2pt of D2c] (D2d) {};
\node [draw, rectangle, fit=(D2a)(D2b)(D2c)(D2d)] (D2) {};
\node (D'2a) at (5, 0.5) {$D'$};
\node [draw, circle, inner sep=0pt, minimum size=5pt, below=0cm of D'2a] (D'2b) {};
\node [draw, circle, inner sep=0pt, minimum size=5pt, below=2pt of D'2b] (D'2c) {};
\node [draw, circle, inner sep=0pt, minimum size=5pt, below=2pt of D'2c] (D'2d) {};
\node [draw, rectangle, fit=(D'2a)(D'2b)(D'2c)(D'2d)] (D'2) {};
\node [draw, rectangle, label=center:$C$] (C2) at (4, -1.4) {\phantom{$A'$}};
\node [draw, rectangle, label=center:$A$] (A2) at (3, -2.8) {\phantom{$A'$}};
\node [draw, rectangle, label=center:$A'$] (A'2) at (5, -2.8) {\phantom{$A'$}};
\draw [thick] (D2) -- (C2) -- (A2);
\draw [thick] (D'2) -- (C2) -- (A'2);
\draw [very thick, color=teal, ->] (D2b) -- (D'2b); 
\draw [very thick, color=teal, ->] (D2c) -- (D'2c); 
\draw [very thick, color=teal, ->] (D2d) -- (D'2d); 

\draw [vecArrow] (6, -1.4) -- (7, -1.4);

\node [draw, rectangle, label=center:$D'$] (D'3) at (8, 0) {\phantom{$A'$}};
\node [draw, rectangle, label=center:$C$] (C3) at (8, -1.4) {\phantom{$A'$}};
\node [draw, rectangle, label=center:$A$] (A3) at (8, -2.8) {\phantom{$A'$}};
\draw [thick] (D'3) -- (C3) -- (A3);

\draw [very thick, color=teal,  ->] (A'2.north) to [bend right] (C2.east);
\draw [very thick, color=teal, ->] (A'2.east) to [bend right] (D'2.east);
\draw [very thick, color=teal, ->] (A'2.west) to [bend left] (A2.east);
\end{tikzpicture}
\end{figure}

\begin{lemma}
\label{lem:exclude-trivial-doublings}
Let $\mathfrak{C}_{12}$ be constructed from some bipartite $H_0$
by one doubling and one collapse, as above.
If $B = \emptyset$ and $E(A, D) = \emptyset$,
then $H_0$ is collapsible onto $\mathfrak{C}_{12}$.
\end{lemma}

\begin{proof}
Let the two steps in a construction of $\mathfrak{C}_{12}$ be such as in
the statement. Recall that $G(S)$ denotes the induced graph on a vertex
set $S$. Note that $H_0 = G(A \cup C \cup D \cup E \cup F$)
and that $\mathfrak{C}_{12} = G(A \cup C \cup D')$. 
For the following discussion
cf.~Figures \ref{fig:last-two-steps} and \ref{fig:auto-a-d-collapse}.

Since $E(A, D) = \emptyset$, the graphs $G(A \cup C \cup D')$ and 
$G(A \cup C \cup D)$ are naturally isomorphic. Therefore, it is enough to show
that it is possible to collapse $E \cup F$ onto $A \cup C \cup D$.

Let us write the collapse that produces $\mathfrak{C}_{12}$ as a homomorphism
$f': A' \cup D \cup E \cup F \cup  F' \to A \cup C \cup D'$.
By Lemma \ref{lem:auto-a-d-collapse}, either $A'$ collapses naturally onto $A$
or $D$ collapses naturally onto $D'$. 

Consider first that $A'$ collapses naturally.
We create a collapsing homomorphism
$f: E \cup F \to A \cup C \cup D$ as follows:
\begin{itemize}
\item If $u \in E$ and $f'(u) \in A \cup C$, then $f(u) := f'(u)$. If 
$f'(u) = w' \in D'$, then $f(u) := w \in D$.
\item For $u \in F$ with $u' \in F'$, if $f'(u') \in A \cup C$, then
$f(u) := f'(u')$. If $f'(u') = w' \in D'$, then $f(u) := w \in D$.
\end{itemize}
We need to see that $f$ is indeed a homomorphism, i.e., that all edges
that touch $E \cup F$ are mapped onto edges of $G(A \cup C \cup D)$.
To this end we make a case analysis:
\begin{itemize}
\item Since $G(E \cup F)$ is naturally isomorphic to $G(E \cup F')$
and $G(A \cup C \cup D)$ is naturally isomorphic to $G(A \cup C \cup D')$,
the edges from $E(E \cup F, E \cup F)$ are preserved by $f$.
\item Since $G(A \cup C \cup D \cup E)$ is naturally isomorphic
to $G(A \cup C \cup D' \cup E)$, the edges from
$E(E, A \cup C \cup D)$ are also preserved by $f$.
\item Let $u \in F$, $v \in A$, $u \sim v$. Then
$u' \sim v' \implies f'(u') \sim v \implies f(u) \sim v$,
where we used that $A'$ collapses naturally.
\item Let $u \in F$, $v \in C$, $u \sim v$. Then
$u' \sim v \implies f'(u') \sim v \implies f(u) \sim v$.
\item Finally, let $u \in F$, $v \in D$, $u \sim v$.
Then $u' \sim v' \implies f'(u') \sim v' \implies f(u) \sim v$.
\end{itemize}

Second, assume that $D$ collapses naturally onto $D'$.
In that case we give a collapsing homomorphism 
$f: \EE \cup F \to A \cup C \cup D'$ as follows:
if $f'(u) \in A \cup C$, then $f(u) := f'(u)$.
If $f'(u) = w' \in D'$, then $f(u) := w \in D$.
To see that $f$ is a collapsing homomorphism, consider:
\begin{itemize}
\item Since $G(A \cup C \cup D)$ is naturally
isomorphic to $G(A \cup C \cup D')$,
$f$ preserves the
edges from $E(E \cup F, A \cup C \cup E \cup F)$.
\item If $u \in E \cup F$, $v \in D$, $u \sim v$
consider the subcases (in all of them we use
that $D$ collapses naturally):
\begin{itemize}
\item If $f'(u) \in A$, then $A \ni f'(u) \sim f'(v) = v' \in D'$,
implying $E(A, D') \ne \emptyset$, a contradiction.
\item If $f'(u) \in C$, then 
$C \ni f(u) = f'(u) \sim f'(v) = v' \implies f(u) \sim v$.
\item If $f'(u) \in D'$, then $f'(u) \sim f'(v) = v' \implies f(u) \sim v$. 
\end{itemize}
\end{itemize}
\end{proof}

\subsection{Putting things together}

\begin{proof}[Proof of Theorem~\ref{thm:cycle-non-constructible}]
By Lemma~\ref{lem:construction-collapsible}, if $\mathfrak{C}_{12}$ is
constructible,
there exists a construction of it by one doubling and one collapse 
starting from some $H_0$ that is not collapsible onto $\mathfrak{C}_{12}$ in the
first place. But this is impossible by Lemmas~\ref{lem:only-trivial-doublings}
and~\ref{lem:exclude-trivial-doublings}.
\end{proof}

\begin{remark}
Our analysis, except for the computer-assisted part, does not depend
on the number of vertices in $\mathfrak{C}_{n}$. Further program runs
confirmed that also $\mathfrak{C}_{14}$ and $\mathfrak{C}_{16}$ are not constructible.
On the other hand, one can see that 
$\mathfrak{C}_8$ and $\mathfrak{C}_{10}$ are constructible.
\end{remark}

\clearpage
\bibliographystyle{alpha}
\bibliography{biblio}
\clearpage
\appendix

\section{Listings of Computer-Assisted Proofs} \label{ch:appendix}

\lstset{
  language=C++,
  numbers=left,
  numberstyle=\tiny,
  frame=single
}

Here we provide program codes for the computer-assisted proofs from
Section \ref{sec:non-constructible}. The programs are written in C++.

\begin{lstlisting}[caption={\texttt{construction.hpp} --- The header file
used by all programs.}, label={lst:header}]
#include <cassert>
#include <climits>
#include <cstdio>
#include <cstdlib>
#include <algorithm>
#include <vector>
using namespace std;

// Bitshifts have higher priority than comparisons.
// Comparisons have higher priority than bit operations.

// Mathematical modulo.
// Precondition: MOD > 0
inline int mod(int x, int MOD) {
  x %= MOD;
  return x + (x < 0 ? MOD : 0);
}

// Number of bits set to one in u.
struct PopCounter {
  int pcnt[1<<16];
  PopCounter() {
    assert(CHAR_BIT == 8 && sizeof(unsigned) == 4);
    for (int i = 1; i < 1<<16; ++i)
      pcnt[i] = pcnt[i/2] + i%2; 
  }
} P;
inline int popcount (unsigned u) {
  return P.pcnt[u & ((1<<16)-1)] + P.pcnt[u >> 16];
}

// Undirected graph with set of vertices S.
// Invariant: all edges inside S.
struct Graph {
  unsigned S;
  vector<unsigned> M;

  Graph(const unsigned a_S, 
      const vector<unsigned>& a_M):
    S(a_S), M(a_M) { }
};

// Doubles subset S of V(G).
// between[uprim] & (1<<u) indicates edge between 
// u' in V(G') and u in V(G). 
// Precondition: T is a subset of G.S
void double_graph(const unsigned T, const Graph& G, 
    Graph& Gprim, vector<unsigned>& between) {
  const vector<unsigned>& M = G.M;
  vector<unsigned>& Mprim = Gprim.M;
  const int N = M.size();
  
  Mprim.resize(N);
  between.resize(N);
  Gprim.S = T;

  for (int u = 0; u < N; ++u)
    if (1<<u & T) {
      Mprim[u] = M[u] & T;
      between[u] = M[u] & ~T;
    } else Mprim[u] = between[u] = 0;
}

// Exchange vertices in T between G and G'.
// Precondition: T is a subset of G.S \cap Gprim.S
void exchange(const unsigned T, Graph& G, Graph& Gprim,
    vector<unsigned>& between) {
  vector<unsigned>& M = G.M;
  vector<unsigned>& Mprim = Gprim.M;
  const int N = M.size();

  for (int u = 0; u < N; ++u) if (1<<u & T) {
    unsigned old_M = M[u], old_Mprim = Mprim[u],
      old_between = between[u];

    M[u] = old_between & ~(1<<u);
    for (int v = 0; v < N; ++v) {
      M[v] &= ~(1<<u);
      if (M[u] & 1<<v) M[v] |= 1<<u;
    }

    // It is important that `between' has not been 
    // modified yet.
    Mprim[u] = 0;
    for (int vprim = 0; vprim < N; ++vprim) 
      if (u != vprim) {
        Mprim[vprim] &= ~(1<<u);
        if (between[vprim] & 1<<u) {
	  Mprim[u] |= 1<<vprim;
	  Mprim[vprim] |= 1<<u;
        }
      }

    between[u] = old_M;
    if (old_between & 1<<u) between[u] |= 1<<u;
    for (int vprim = 0; vprim < N; ++vprim) 
      if (u != vprim) {
        between[vprim] &= ~(1<<u);
        if (old_Mprim & 1<<vprim) 
          between[vprim] |= 1<<u;
      }
  }
}

// Can Gprim be collapsed onto G?
// If yes, `mapping' will contain a mapping 
// from Gprim to G,  with mapping[uprim] == -1 
// for uprim not in Gprim.S. 
bool is_collapsible(const Graph& a_G, 
    const Graph& a_Gprim, 
    const vector<unsigned>& a_between, 
    vector<int>& a_mapping) {
  struct RecursiveData {
    const Graph& G;
    const Graph& Gprim;
    const vector<unsigned>& between;
    vector<int>& mapping;
    const vector<unsigned>& M;
    const vector<unsigned>& Mprim;
    const int N;

    RecursiveData(const Graph& a_G, 
        const Graph& a_Gprim, 
        const vector<unsigned>& a_between, 
        vector<int>& a_mapping):
	  G(a_G), Gprim(a_Gprim), between(a_between),
          mapping(a_mapping), M(G.M), Mprim(Gprim.M), 
          N(M.size()) {
      mapping.resize(N);
      fill_n(mapping.begin(), N, -1);
    }

    bool is_collapsible_rec(int uprim) {
      if (uprim == N) return true;
      if (1<<uprim & ~Gprim.S) 
        return is_collapsible_rec(uprim+1);
      // invariant: u' < N and u' in V(G')

      for (int u = 0; u < N; ++u) if (1<<u & G.S) {
	if ((M[u] & between[uprim]) != between[uprim]) 
          continue;
	// invariant: u' -> u preserves edges between 
        // u' and G

	bool ok = true;
	for (int vprim = 0; vprim < uprim && ok; 
            ++vprim) {
	  if (Mprim[uprim] & 1<<vprim && 
              !(M[u] & 1<<mapping[vprim])) {
	    ok = false;
          }
        }
	if (!ok) continue;
	// invariant: u' -> u preserves edges between 
        // u' and preceding vertices in G'

	mapping[uprim] = u;
	if (is_collapsible_rec(uprim+1)) return true;
      }
      return false;
    }
  } R(a_G, a_Gprim, a_between, a_mapping);

  return R.is_collapsible_rec(0);
}

// Can G' be collapsed onto G such that both T and 
// G'.S \setminus T do not collapse naturally?
// If yes, mapping will contain such mapping from 
// G' to G, with mapping[u'] == -1 for 
// u' not in G'.S.
// Precondition: T is a subset of G'.S which is 
// a subset of G.S
bool is_unnaturally_collapsible(const unsigned a_T, 
    const Graph& a_G, const Graph& a_Gprim, 
    const vector<unsigned>& a_between,
    vector<int>& a_mapping) {
  struct RecursiveData {
    const unsigned T;
    const Graph& G;
    const Graph& Gprim;
    const vector<unsigned>& between;
    vector<int>& mapping;
    const vector<unsigned>& M;
    const vector<unsigned>& Mprim;
    const int N;

    RecursiveData(const unsigned a_T, const Graph& a_G, 
        const Graph& a_Gprim, 
        const vector<unsigned>& a_between, 
        vector<int>& a_mapping):
	  T(a_T), G(a_G), Gprim(a_Gprim), 
          between(a_between), mapping(a_mapping), 
          M(G.M), Mprim(Gprim.M), N(M.size()) {
      mapping.resize(N);
      fill_n(mapping.begin(), N, -1);
    }

    bool is_collapsible_rec(int uprim) {
      if (uprim == N) {
	bool ok1 = false, ok2 = false;
	for (int uprim = 0; uprim < N && (!ok1 || !ok2); 
            ++uprim) {
	  if (1<<uprim & ~Gprim.S) continue;
	  if (1<<uprim & T && mapping[uprim] != uprim)
            ok1 = true;
	  else if (1<<uprim & ~T && 
              mapping[uprim] != uprim) {
            ok2 = true;
          }
	}
	return ok1 && ok2;
      }

      if (1<<uprim & ~Gprim.S) 
        return is_collapsible_rec(uprim+1);
      // invariant: u' < N and u' in V(G')

      for (int u = 0; u < N; ++u) if (1<<u & G.S) {
	if ((M[u] & between[uprim]) != between[uprim]) 
          continue;
	// invariant: u' -> u preserves edges between 
        // u' and G

	bool ok = true;
	for (int vprim = 0; vprim < uprim && ok; 
            ++vprim) {
	  if (Mprim[uprim] & 1<<vprim && 
              !(M[u] & 1<<mapping[vprim])) {
	    ok = false;
          }
        }
	if (!ok) continue;
	// invariant: u' -> u preserves edges between 
        // u' and preceding vertices in G'.

	mapping[uprim] = u;
	if (is_collapsible_rec(uprim+1)) return true;
      }
      return false;
    }
  } R(a_T, a_G, a_Gprim, a_between, a_mapping);

  return R.is_collapsible_rec(0);
}

// Precondition: T is a subset of G.S
inline unsigned neighbors (const unsigned T, 
    const Graph& G) {
  const int N = G.M.size();
  unsigned res = 0;
  for (int u = 0; u < N; ++u) if (1<<u & T)
    res |= G.M[u];
  return res;
}

const int V = 12;
// Cycle with shortcuts C_V.
Graph original_G() {
  Graph G((1<<V) - 1, vector<unsigned>(V));
  for (int u = 0; u < V; ++u) 
    for (int s = -3; s <= 3; s += 2)
      G.M[u] |= 1 << mod(u+s, V);
  return G;
}
\end{lstlisting}

\begin{lstlisting}[caption={\texttt{non\_empty\_b.cpp} --- Proof
of Lemma \ref{lem:assisted-non-empty-b}.},
label={lst:non-empty-b}]
#include "construction.hpp"

// Precondition: B, C disjoint, 0 in B
bool Bprim_filled(const unsigned a_C, 
    const unsigned a_B) {
  struct RecData {
    const Graph G;
    const vector<unsigned>& M;
    const unsigned C;
    const unsigned B;
    unsigned Bprim;
    const int pB;
    vector<int> B_list, Bprim_list;

    RecData(const unsigned a_C, const unsigned a_B):
        G(original_G()), M(G.M), C(a_C), B(a_B), 
        Bprim(0), pB(popcount(B)), B_list(pB), 
        Bprim_list(pB) {
      for (int u = 0, ind = -1; u < V; ++u)
	if (1<<u & B) {
	  ++ind;
	  B_list[ind] = u;
	}
    }

    bool recursively_filled(int ind) {
      if (ind == pB) {
	// invariant: B, B', C (pairwise) disjoint
	// invariant: edges of B and B' (inside and 
        // to C) isomorphic according to Bprim_list.
	return is_rest_filled();
      }

      const int u = B_list[ind];
      for (int uprim = 0; uprim < V; ++uprim) {
	if (1<<uprim & (B|C|Bprim)) continue;
	// invariant: uprim is "fresh"
	if (M[uprim] & B) continue;
	// invariant: no edges to B
	if ((M[u]&C) != (M[uprim]&C)) continue;
	// invariant: edges to C the same
	bool ok = true;
	for (int j = 0; j < ind && ok; ++j) {
	  const int v = B_list[j], 
            vprim = Bprim_list[j];
	  // a hack: `!' is used to convert to bool
	  if (!(M[v]&(1<<u)) != !(M[vprim]&(1<<uprim))) 
            ok = false;
	}
	if (!ok) continue;
	// invariant: edges inside B and B' (so far) 
        // isomorphic
	Bprim |= 1<<uprim;
	Bprim_list[ind] = uprim;
	if (recursively_filled(ind+1)) return true;
	Bprim &= ~(1<<uprim);
      }
      return false;
    }

    // preconditions: B, Bprim, C disjoint
    // B_list, Bprim_list, pB correctly filled
    // B and B' isomorphic wrt each other and C
    bool is_rest_filled() {
      static Graph Gout(0, vector<unsigned>(V));
      static vector<unsigned> between(V);
      static vector<unsigned>& Mout = Gout.M;
      static vector<int> mapping(V);
      
      for (unsigned A = 0; A < 1<<V; A += 2) {
	if (A & (B|C|Bprim)) continue;
	// invariant: A, B, B', C disjoint
	if (neighbors(A, G) & (Bprim)) continue;
	// invariant: no edges between A and B'

	const unsigned Dprim = ((1<<V)-1) & 
          ~(A|B|Bprim|C);
	if (neighbors(Dprim, G) & (A|B)) continue;
	// invariant: no edges between D' and A \cup B
	
	Gout.S = A|Dprim;
	for (int u = 0; u < V; ++u)
	  if (1<<u & A) {
	    Mout[u] = M[u] & A;
	    between[u] = M[u] & C;
	    for (int ind = 0; ind < pB; ++ind) {
	      const int v = B_list[ind], 
                vprim = Bprim_list[ind];
	      if (M[u] & 1<<v) between[u] |= 1<<vprim;
	    }
	  } else if (1<<u & Dprim) {
	    Mout[u] = M[u] & Dprim;
	    between[u] = M[u] & C;
	    for (int ind = 0; ind < pB; ++ind) {
	      const int v = B_list[ind], 
                vprim = Bprim_list[ind];
	      if (M[u] & 1<<vprim) between[u] |= 1<<v;
	    }
	  } else Mout[u] = between[u] = 0;

	if (is_collapsible(G, Gout, between, mapping)) {
	  printf("FAILURE\nA = ");
	  for (int u = 0; u < V; ++u) if (1<<u & A) 
            printf("%d ", u);
	  printf("\n(B,B') = ");
	  for (int ind = 0; ind < pB; ++ind)
	    printf("(%d, %d) ", B_list[ind], 
              Bprim_list[ind]);
	  printf("\nC = ");
	  for (int u = 0; u < V; ++u) if (1<<u & C) 
            printf("%d ", u);
	  printf("\nD' = ");
	  for (int u = 0; u < V; ++u) if (1<<u & Dprim) 
            printf("%d ", u);
 	  printf("\nmapping = ");
	  for (int u = 0; u < V; ++u)
	    printf("(%d->%d) ", u, mapping[u]);
	  printf("\n");
	  exit(0);
	}
      }
      return false;
    }
  } R(a_C, a_B);

  return R.recursively_filled(0);
}

// Assume E(A, D) is empty.
// Try all partitions of C_12 into A, B, B', C, D' 
// s.t. in the last doubling:
// A is doubled and then A is fixed and A' collapsed.
// (non-empty) B is doubled and fixed together with B'.
// C is fixed in both steps.
// D is doubled, with D collapsed and D' fixed.
// Objective: show that resulting A', D cannot be 
// collapsed onto the rest.
int main() {
  printf("non-empty B, V = %d\n", V);
  // Assume w.l.o.g. that 0 is in B.
  for (unsigned C = 0; C < 1<<V; C += 2) 
    for (unsigned B = 1; B < 1<<V; B += 2) {
      // invariant: 0 in B
      if (B&C || popcount(B)%2 == 1) continue;
      // invariant: B, C disjoint
      if (Bprim_filled(C, B)) {
        // this should be never executed
        printf("INTERNAL ERROR\n");
        exit(1);
      }
    }
  printf("SUCCESS\n");
}
\end{lstlisting}

\begin{lstlisting}[caption={\texttt{non\_empty\_ad.cpp} --- Proof
of Lemma \ref{lem:assisted-non-empty-ad}.},
label={lst:non-empty-ad}]
#include "construction.hpp"

// Assume B is empty and |E(A, D)| = 1.
// Try all partitions of C_12 into A, C, D' s.t. in the 
// last doubling:
// A is doubled and A' is later collapsed.
// C is not doubled.
// D is doubled and later collapsed and D' is kept.
// Then try all choices for the edge between A and D.
// Goal: Show that resulting A', D  cannot be collapsed
// onto A, C, D'.
int main() {
  printf("|E(A,D)| = 1, V = %d\n", V);
  Graph G = original_G();
  for (unsigned A = 0; A < 1<<V; ++A) 
    for (unsigned C = 0; C < 1<<V; ++C) {
      if (A&C) continue;
      // invariant: A, C disjoint
      unsigned Dprim = ((1<<V)-1) & ~(A|C);
      if (neighbors(Dprim, G) & A) continue;
      // invariant: no edges between A and D'

      Graph tmp_G = G, Gprim(0, vector<unsigned>());
      vector<unsigned> between;
      vector<int> mapping;

      double_graph(A|Dprim, tmp_G, Gprim, between);
      exchange(Dprim, tmp_G, Gprim, between);
    
      for (int u = 0; u < V; ++u) if (1<<u & A)
        for (int v = 0; v < V; ++v) if (1<<v & Dprim) {
	  if (u%2 == v%2) continue;
          // invariant: u and v do not create odd cycle
          between[u] |= 1<<v;
          between[v] |= 1<<u;
          if (is_collapsible(tmp_G, Gprim, between,
              mapping)) {
            printf("FAILURE\nA = ");
            for (int w = 0; w < V; ++w) if (1<<w & A) 
              printf("%d ", w);
            printf("\nC = ");
            for (int w = 0; w < V; ++w) if (1<<w & C) 
              printf("%d ", w);
            printf("\nDprim = ");
            for (int w = 0; w < V; ++w) 
              if (1<<w & Dprim) printf("%d ", w);
            printf("\nu = %d v = %d\nmapping = ", u, v);
            for (int w = 0; w < (int)mapping.size(); 
                ++w) {
              printf("(%d -> %d) ", w, mapping[w]);
            }
            printf("\n");
            exit(0);
          }
          between[u] &= ~(1<<v);
          between[v] &= ~(1<<u);
        }
    }
  printf("SUCCESS\n");
}
\end{lstlisting}

\begin{lstlisting}[caption={\texttt{natural\_collapse.cpp} --- Proof
of Lemma \ref{lem:assisted-a-d-collapse}.},
label={lst:natural}]
#include "construction.hpp"

// Assume E(A, D) is empty.
// Try partitioning vertices of C_12 into A, C, D' s.t.
// in the last doubling:
// A is doubled and A' is later collapsed.
// C is not doubled.
// D is doubled and later collapsed and D' is kept.
// Objective: Show that every time either A' or D must 
// be naturally collapsed.
int main() {
  printf("Natural collapse lemma, V = %d\n", V);
  Graph G = original_G();
  for (unsigned A = 0; A < 1<<V; ++A) 
    for (unsigned C = 0; C < 1<<V; ++C) {
      if (A&C) continue;
      // invariant: A, C disjoint
      unsigned D = ((1<<V)-1) & ~(A|C);
      if (neighbors(D, G) & A) continue;
      // invariant: no edges between A and D

      Graph tmp_G = G, Gprim(0, vector<unsigned>());
      vector<unsigned> between;
      vector<int> mapping;
    
      double_graph(A|D, tmp_G, Gprim, between);
      exchange(D, tmp_G, Gprim, between);
      if (is_unnaturally_collapsible(A, tmp_G, Gprim, 
          between, mapping)) {
        printf("FAILURE\nA = ");
        for (int u = 0; u < V; ++u) 
          if (1<<u & A) printf("%d ", u);
        printf("\nC = ");
        for (int u = 0; u < V; ++u) 
          if (1<<u & C) printf("%d ", u);
        printf("\nD = ");
        for (int u = 0; u < V; ++u) 
          if (1<<u & D) printf("%d ", u);
        printf("\nmapping = ");
        for (int u = 0; u < (int)mapping.size(); ++u)
          printf("(%d -> %d) ", u, mapping[u]);
        printf("\n");
        exit(0);
      }
    }
  printf("SUCCESS\n");
}
\end{lstlisting}

\end{document}